\documentclass[10pt,conference]{IEEEtran}
\IEEEoverridecommandlockouts

\newif\ifrevisioncover
\revisioncoverfalse

\usepackage{cite}
\usepackage{amsmath,amssymb,amsfonts,amsthm}
\usepackage{graphicx}
\usepackage[caption=false,font=footnotesize]{subfig}
\usepackage{float}
\usepackage{booktabs}
\usepackage{tabularx}
\usepackage{textcomp}
\usepackage{xcolor}
\usepackage{xspace}
\usepackage{enumitem}
\usepackage{microtype}
\usepackage[hyphens]{url}
\usepackage{fancyhdr}
\usepackage{hyperref}
\hypersetup{hidelinks}

\providecommand{\sysname}{Twig}
\providecommand{\papertitle}{%
  \sys: Parallel and Memory-Efficient Tree Speculative Decoding for
  Hybrid-Attention LLMs}

\providecommand{\motivationfigdir}{figures/motivation/output}
\providecommand{\evaluationfigdir}{figures/evaluation/output}

\newif\ifrevisioncolors
\revisioncolorsfalse
\newif\ifrevisionlabels
\revisionlabelsfalse
\input{macros}

\newcommand{\hpcayear}{2027}

\newcommand{\hpcasubmissionnumber}{814}
\title{\papertitle}

\def\hpcacameraready{} 

\newcommand{\hpcaauthors}{%
  Li Wang\IEEEauthorrefmark{1}\IEEEauthorrefmark{2},
  Yi Su\IEEEauthorrefmark{1},
  Xiabao Wu\IEEEauthorrefmark{2},
  Chiran You\IEEEauthorrefmark{1},
  Yongchao Liu\IEEEauthorrefmark{2},
  Zhan Qiu\IEEEauthorrefmark{2},
  Juelu Zhang\IEEEauthorrefmark{2},
  Jiajun Zheng\IEEEauthorrefmark{2},\\
  Fangxin Liu\IEEEauthorrefmark{3},
  Jie Zhang\IEEEauthorrefmark{4},
  Chen Tian\IEEEauthorrefmark{1},
  and Chengying Huan\IEEEauthorrefmark{1}%
  \thanks{Corresponding authors: Chengying Huan and Yongchao Liu.}%
}
\newcommand{\hpcaaffiliation}{%
  \IEEEauthorrefmark{1}State Key Laboratory for Novel Software Technology,
  Nanjing University
  \quad \IEEEauthorrefmark{2}Ant Group\\
  \quad \IEEEauthorrefmark{3}Shanghai Jiao Tong University
  \quad \IEEEauthorrefmark{4}Peking University%
}
\newcommand{\hpcaemail}{%
  \footnotesize
  \{lwang.cs,suyi\}@smail.nju.edu.cn,
  wuxiabao.wxb@antgroup.com,
  231840050@smail.nju.edu.cn\\
  \{yongchao.ly,qiuzhan.qz,zhangjuelu.zjl\}@antgroup.com,
  zhengjiajun@gmail.com\\
  liufangxin@sjtu.edu.cn,
  jiez@pku.edu.cn,
  \{tianchen,huanchengying\}@nju.edu.cn%
}

\ifrevisioncover
  \AtBeginDocument{%
    \pagenumbering{Roman}
    \input{sections/cover-letter}%
    \clearpage
    \pagenumbering{arabic}
  }
\fi

\author{
  \ifdefined\hpcacameraready
    \IEEEauthorblockN{\hpcaauthors{}}
      \IEEEauthorblockA{
        \hpcaaffiliation{} \\
        \hpcaemail{}
      }
  \else
    \IEEEauthorblockN{\normalsize{HPCA \hpcayear{} Submission
      \textbf{\#\hpcasubmissionnumber{}}} \\
      \IEEEauthorblockA{
        Confidential Draft \\
        Do NOT Distribute!!
      }
    }
  \fi 
}

\fancypagestyle{camerareadyfirstpage}{%
  \fancyhead{}
  
  \fancyhead[C]{
    \ifdefined\aeopen
    \parbox[][12mm][t]{13.5cm}{\hpcayear{} IEEE International Symposium on High-Performance Computer Architecture (HPCA)}    
    \else
      \ifdefined\aereviewed
      \parbox[][12mm][t]{13.5cm}{\hpcayear{} IEEE International Symposium on High-Performance Computer Architecture (HPCA)}
      \else
      \ifdefined\aereproduced
      \parbox[][12mm][t]{13.5cm}{\hpcayear{} IEEE International Symposium on High-Performance Computer Architecture (HPCA)}
      \else
      \parbox[][0mm][t]{13.5cm}{\hpcayear{} IEEE International Symposium on High-Performance Computer Architecture (HPCA)}
    \fi 
    \fi 
    \fi 
    \ifdefined\aeopen 
      \includegraphics[width=12mm,height=12mm]{ae-badges/open-research-objects.pdf}
    \fi 
    \ifdefined\aereviewed
      \includegraphics[width=12mm,height=12mm]{ae-badges/research-objects-reviewed.pdf}
    \fi 
    \ifdefined\aereproduced
      \includegraphics[width=12mm,height=12mm]{ae-badges/results-reproduced.pdf}
    \fi
  }
  \fancyfoot[C]{}
}
\begin{document}
\maketitle

\ifdefined\hpcacameraready 
  \thispagestyle{camerareadyfirstpage}
  \pagestyle{empty}
\else
  \thispagestyle{plain}
  \pagestyle{plain}
\fi

\newcommand{\hpcaheight}{0mm}
\ifdefined\eaopen
\renewcommand{\hpcaheight}{12mm}
\fi

\thispagestyle{empty}


\begin{abstract}
Emerging LLMs increasingly interleave full attention with recurrent linear attention to reduce long-context serving costs. Tree speculative decoding could further accelerate these hybrid models, but recurrent linear-attention layers carry parent-dependent states. Existing engines therefore verify proposal trees sequentially and snapshot the full state per node, causing latency and memory to grow with tree and batch sizes. We present \sys, a kernel--runtime co-design for tree speculative decoding of hybrid LLMs. \sys derives a tree-structured factorization of the linear-attention recurrence and realizes it with a resource-efficient GPU kernel that verifies all nodes in parallel, accelerating linear-attention tree verification by 3.4--7.7$\times$. It losslessly represents speculative state updates as token-level factors and reconstructs only the accepted state, reducing transient state memory by 82--99$\times$ and freeing GPU memory for KV caches. We integrate \sys into SGLang, a widely deployed production LLM serving engine, and further optimize its runtime through batch-wide verification budgeting. Across four models, two GPU platforms, and diverse datasets, \sys achieves up to $4.72\times$ the offline decode throughput of autoregressive decoding and up to $2.03\times$ that of the strongest tree-speculative baseline. Under online agent workloads, \sys reduces TTFT and TPOT by up to $67.6\%$ and $49.9\%$, respectively, over the strongest tree-speculative baseline.
\end{abstract}

\section{Introduction}
\label{sec:introduction}

Large language models (LLMs) increasingly serve long-context~\cite{zhiqiang2026Strata} and agentic
workloads~\cite{michael2026agentix, yang2026cacheslide}, where decoding long responses is a major serving cost~\cite{kan2025nanoflow, amey2024sarathiserve}. Full
attention computes against and retains a key--value (KV) pair for every
preceding token~\cite{ruoyu2025mooncake}, so its computation, data movement, and memory footprint grow
with context length~\cite{shutian2026directkv, yi2026pat}. Models such as Qwen3.5~\cite{qwen2026qwen35} and Kimi Linear~\cite{kimi2025linear} instead interleave
full attention with recurrent linear attention~\cite{yang2025gateddeltanet}. Linear-attention layers summarize the prefix in a fixed-size state~\cite{pan2025Marconi}, while full-attention layers preserve global retrieval.
Hybrid attention thereby reduces the per-token serving cost~\cite{Jung2025hlx}, but generation
still advances through one target-model invocation per output token.

Speculative decoding~\cite{leviathan2023speculative, chen2023accelerating}
reduces these serial invocations by exploiting unused compute capacity during
memory-bound decoding~\cite{li2025eagle3, liu2026speculative}. Among
speculative decoding methods, tree speculation is particularly effective due
to its high average acceptance length: EAGLE-2 reports 4--5.5 accepted tokens
per drafting--verification cycle and up to $5\times$
speedup~\cite{li2024eagle2}. A lightweight drafter proposes a tree
of candidate tokens, and the target model verifies all nodes together in one
invocation. For full-attention models, an ancestor mask allows
each proposal node to attend to the committed prefix and its own ancestors, so
the target model can score the entire tree in one parallel
forward~\cite{miao2024specinfer, yao2025deft} before sampling commits an
accepted path. Accepted candidates amortize model-weight loading and invocation
overheads across multiple tokens while preserving the target model's
output distribution~\cite{leviathan2023speculative}.

Tree speculation is therefore a promising approach to accelerating
hybrid-attention LLMs. However, realizing this potential requires both an
efficient GPU \textbf{kernel} for executing proposal trees over recurrent linear
attention and a serving \textbf{runtime} that manages speculative recurrent states and
tree work across requests. Existing systems provide neither capability for
hybrid-attention LLMs. We target this kernel--runtime gap.

\textbf{Existing tree-speculation kernels cannot efficiently parallelize
modern hybrid-attention verification.} SpecInfer~\cite{miao2024specinfer} and DeFT~\cite{yao2025deft} flatten a proposal tree
into an ancestor-masked sequence and score all nodes in one full-attention
forward. These full-attention methods do not transfer directly to recurrent
linear attention because its recurrent formulation exposes no attention matrix
to mask. Instead, it encodes tree dependencies through parent-to-child state
propagation.
STree~\cite{wu2025stree} composes diagonal state-space model (SSM) transitions over proposal
trees, but this composition does not apply to modern gated
delta recurrences~\cite{yang2025gateddeltanet, qwen2026qwen35}.

\textbf{Existing tree-speculation runtimes inefficiently manage hybrid
recurrent states and execution costs.} Existing systems optimize draft policies
and tree construction~\cite{hu2026gto}, tree-aware serving
pipelines~\cite{miao2024specinfer, guan2025yggdrasil}, and service-level
objective (SLO)-aware budgeting~\cite{li2026adaserve}. In particular,
SGLang~\cite{zheng2024sglang, GitHubsglang} provides a production
tree-speculation runtime, while AdaServe~\cite{li2026adaserve} adapts tree
budgets to heterogeneous SLOs. Both are designed around the append-only KV
caches of full-attention models.
Hybrid runtimes must manage the large, transient recurrent states
generated for candidate branches during tree verification~\cite{GitHubsglang}
and size tree budgets across heterogeneous target
operators~\cite{Jung2025hlx, li2026adaserve}. Existing systems do neither
efficiently.

\begin{figure*}[t]
  \centering
  \begin{minipage}[t]{0.315\textwidth}
    \centering
    \includegraphics[width=\linewidth]{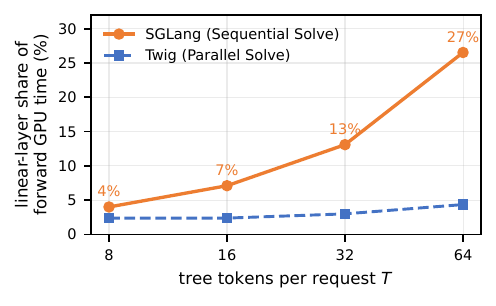}\\[-0.5ex]
    {\small (a) Linear-attention verification time grows rapidly with tree size.}
  \end{minipage}\hfill
  \begin{minipage}[t]{0.315\textwidth}
    \centering
    \includegraphics[width=\linewidth]{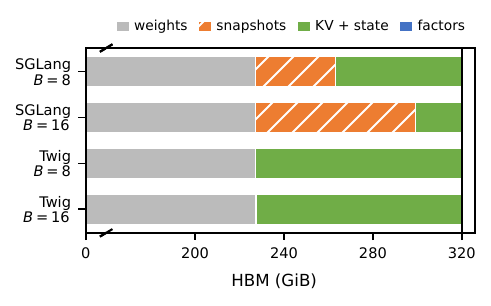}\\[-0.5ex]
    {\small (b) Per-node snapshots consume memory and reduce serving capacity.}
  \end{minipage}\hfill
  \begin{minipage}[t]{0.315\textwidth}
    \centering
    \includegraphics[width=\linewidth]{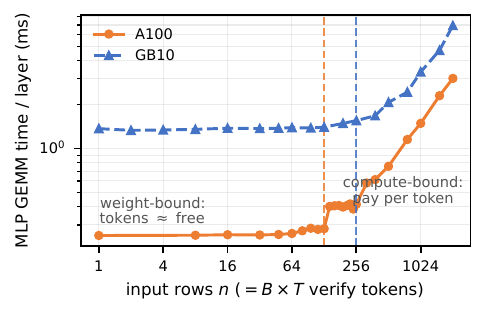}\\[-0.5ex]
    {\small (c) The efficient token range is hardware-specific.}
  \end{minipage}
  \caption{Serial recurrence, snapshot memory overhead, and hardware-dependent
  capacity motivate hybrid-aware tree verification. Measurements use
  Qwen3.5-9B on A100 in (a), Qwen3.5-122B-A10B on $4\times$A100 in (b), and
  Qwen3.5-9B on A100 and GB10 in (c).}
  \label{fig:motivation}
\end{figure*}

In practice, these kernel and runtime limitations manifest as three system
bottlenecks, quantified in \figref{fig:motivation} and detailed in
\secref{sec:motivation}.
\textbf{(1) Serial recurrence.} Parent-to-child execution grows with the
proposal tree, increasing linear attention's share of the target forward from
4\% to 27\% as the tree expands from 8 to 64 nodes.
\textbf{(2) State explosion.} Full-state snapshots grow with both tree and
batch sizes, consuming 36~GiB for eight 32-node requests on
Qwen3.5-122B-A10B. This transient storage displaces KV-cache capacity and
limits batching, ultimately reducing serving throughput.
\textbf{(3) Variable verification capacity.} The number of proposal tokens
that can be verified efficiently changes with the GPU and the current batch.

Overcoming these limitations presents three challenges.
\textbf{(1) Tree-coupled recurrent dependencies.} Each proposal node consumes
its parent's state, imposing a topological critical path that an ancestor mask
cannot remove. Parallel execution must eliminate this order while preserving
the exact recurrence and mapping efficiently onto the GPU.
\textbf{(2) State divergence before branch selection.} Candidate branches
produce distinct large states before the runtime knows which path will be
accepted. Exact recovery thus appears to require materializing a complete state
for every node.
\textbf{(3) Coupled, configuration-dependent cost.} Full-attention, recurrent
linear-attention, and dense operators scale with different request and batch
properties. Because they share one target forward, their joint efficient
operating point depends on the GPU, KV lengths, and tree composition,
precluding a fixed per-request capacity. Together, these challenges span the
linear-attention algorithm, GPU kernel, and serving runtime.

We present \textbf{\sys}, a kernel--runtime co-design that addresses these challenges
through three mechanisms. \textbf{(1) Parallel tree solver.} A tree-structured
factorization separates the committed state from interactions among proposal
tokens, yielding a structured linear system that a resource-efficient GPU
kernel solves for all nodes in parallel from one pre-tree state.
\textbf{(2) Factorized state lifecycle.} \sys retains one committed state and
represents candidate updates with token-scale factors. After sampling, it
reconstructs and commits only the accepted state, while rejected branches
require neither full recurrent states nor rollback.
\textbf{(3) Hardware-aware verification budget.} \sys profiles the efficient
verification capacity of each execution configuration and allocates this
batch-wide budget across request trees. We implement the complete design in
SGLang~\cite{zheng2024sglang}, covering tree construction, target-model
execution, sampling, and state commit.

In summary, our work makes the following contributions:
\begin{itemize}[leftmargin=1.25em,itemsep=1pt,topsep=2pt,parsep=0pt]
  \item We derive an exact tree-structured factorization of the
  linear-attention recurrence and develop a resource-efficient GPU kernel that
  verifies all proposal nodes in parallel, accelerating linear-attention tree
  verification by up to $7.7\times$.
  \item We develop a factorized state lifecycle that losslessly represents
  candidate updates with token-scale factors and reconstructs only the accepted
  state, reducing transient state memory by $82$--$99\times$.
  \item We develop a hardware-aware, batch-wide verification budget and
  integrate \sys into SGLang to support continuous batching, CUDA Graphs, and
  GPU-resident state commit.
  \item Across four models and two GPU platforms, \sys achieves up to
  $4.72\times$ the offline decode throughput of autoregressive decoding and
  $2.03\times$ that of the strongest tree-speculative baseline. On online
  agent workloads, \sys reduces time to first token (TTFT) and time per output
  token (TPOT) by up to
  $67.6\%$ and $49.9\%$, respectively, over the strongest tree-speculative
  baseline.
\end{itemize}

\section{Background and Motivation}

\subsection{Hybrid Attention Models and Gated-Delta Linear Attention}
\label{sec:background-hybrid}

Hybrid LLMs interleave full-attention and recurrent linear-attention layers,
combining long-range retrieval with a bounded recurrent state~\cite{pim2025Pimba}. Qwen3.5, for
example, repeats three Gated DeltaNet (GDN) layers followed by one gated
full-attention layer~\cite{qwen2026qwen35}. \figref{fig:hybrid-background} contrasts their history
representations: full attention retains a KV pair per prefix token and grows
with context length, whereas GDN summarizes the prefix in a fixed-size matrix.

\begin{figure}[H]
  \centering
  \includegraphics[width=\linewidth]{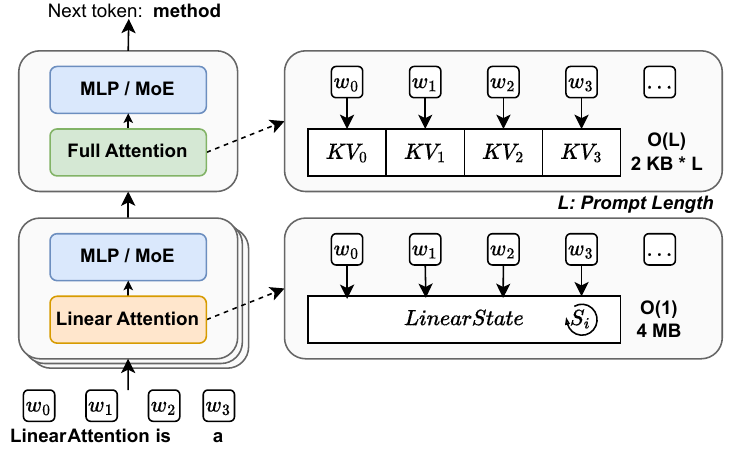}
  \caption{State representations in hybrid-attention models.}
  \label{fig:hybrid-background}
\end{figure}

A GDN layer applies a short causal depthwise convolution and input-dependent
projections and gates to form its recurrent operands. Its recurrent core has
$H$ independent \emph{state heads}, each maintaining a $d_k\times d_v$ matrix.
For one state head, GDN updates~\cite{yang2025gateddeltanet}
$S_t\in\mathbb R^{d_k\times d_v}$ as
\begin{compactequation}
\begin{aligned}
\widetilde S_t&=\gamma_tS_{t-1},&
u_t&=\beta_t(v_t-\widetilde S_t^{\top}k_t),\\
S_t&=\widetilde S_t+k_tu_t^{\top},&
o_t&=S_t^{\top}q_t,
\end{aligned}
\label{eq:gdn-background}
\end{compactequation}
where $q_t,k_t\in\mathbb R^{d_k}$, $v_t,u_t,o_t\in\mathbb R^{d_v}$, and
$\gamma_t>0$ and $\beta_t$ are decay and write gates. The in-place state $S_t$
summarizes the prefix through token $t$. Because $S_t$ depends on $S_{t-1}$,
ordinary decoding evaluates the recurrence in token order.

\subsection{Tree Speculative Decoding}
\label{sec:background-tree}

Autoregressive decoding ordinarily invokes the target model once per output
token~\cite{patel2025splitwise}. At small batches, each invocation is memory-bound: model-weight
transfers dominate while much of the GPU's arithmetic throughput remains idle~\cite{he2025papi}.
Tree speculative decoding exploits this headroom with a lightweight drafter,
such as multi-token prediction (MTP)~\cite{Gloeckle2024better} or
EAGLE~\cite{li2024eagle2}, that constructs a proposal tree, after which the target
scores all nodes together~\cite{miao2024specinfer, ye2025flashinfer}. Reusing each weight load across more token
computations makes verification latency grow much more slowly than the number
of verified tokens, so extra tokens are nearly free within the hardware's
efficient range~\cite{li2024specpim, zhang2026swiftspec}. Standard acceptance and correction sampling can then commit
multiple tokens without changing the target model's output distribution~\cite{chen2023accelerating}.

\begin{figure}[H]
  \centering
  \includegraphics[width=\linewidth]{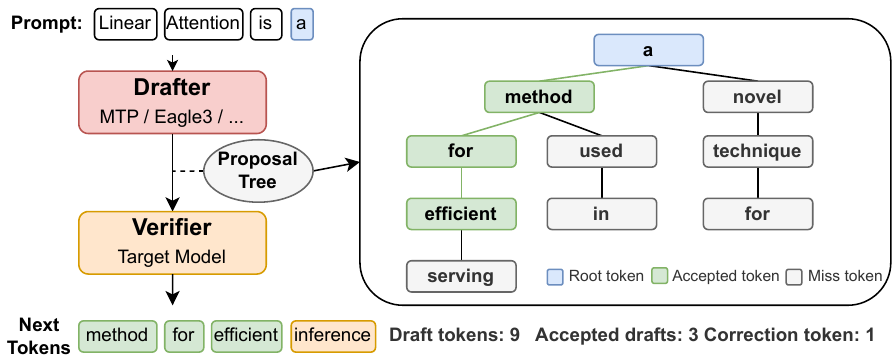}
  \caption{Tree speculative decoding and accepted-path commit.}
  \label{fig:tree-background}
\end{figure}

\figref{fig:tree-background} shows one round. Each proposal node is a draft
token, each root-to-node path is a candidate continuation, and branches reuse
their shared prefix. An \emph{ancestor mask} restricts each node to the
committed prefix and its own proposal ancestors, preventing leakage from
sibling branches. Sampling commits one accepted path and a target-sampled
correction token, then discards the other branches. The accepted-path length
determines per-round progress, and we report its mean as mean accepted tokens
(MAT).

\subsection{Limitations of Existing Tree Verification for Hybrid Attention}
\label{sec:motivation}

Existing work cannot efficiently support tree speculative decoding for
hybrid-attention models.  We next identify three limitations, using
\figref{fig:motivation} as empirical evidence.

\runinhead{L1: Linear-attention verification time grows rapidly with tree size.}
Each proposal node requires the linear-attention state produced by its parent.
Existing engines~\cite{GitHubsglang} must therefore traverse a $T$-node tree and apply
\eqnref{eq:gdn-background} one node at a time, creating $T$ sequential
recurrent steps even though all proposal tokens are available before the
target forward.  \figref{fig:motivation}(a) shows the resulting latency
growth: increasing the tree from 8 to 64 nodes raises the linear layers'
share of the target forward from 4\% to 27\%.

\runinhead{L2: Per-node state snapshots impose prohibitive memory overhead.}
Before sampling, the accepted path is unknown. Existing engines materialize a
complete post-update GDN state for every tree node because one
in-place state cannot represent diverging branches. Sampling commits only the
terminal accepted state and discards every rejected snapshot. Beyond HBM
bandwidth, these snapshots displace KV-cache and active-request capacity~\cite{kwon2023efficient, biao2024llumnix},
growing linearly with batch size $B$ and tree size $T$.
\figref{fig:motivation}(b) quantifies this batch-amplified loss when serving
Qwen3.5-122B-A10B on four A100-80GB GPUs. With $T=32$, snapshots consume
36~GiB at $B=8$ and 72~GiB at $B=16$, while \sys's factors remain below
1~GiB in both cases. For one request with a 100-node tree,
\tabref{tab:transient-memory} reports 4.7--14.1~GB for snapshots but only
57--151~MB for \sys's token-scale factors.

\begin{table}[H]
  \centering
  \caption{Transient state memory per request.}
  \label{tab:transient-memory}
  \small
  \setlength{\tabcolsep}{1.5pt}
  \begin{tabular}{@{}lrrrrr@{}}
    \toprule
    & \multicolumn{1}{c}{Linear} & \multicolumn{2}{c}{SGLang} &
      \multicolumn{2}{c}{\sys} \\
    \cmidrule(lr){3-4}\cmidrule(l){5-6}
    Model & \multicolumn{1}{c}{state} & $T{=}100$ & $T{=}200$ &
      $T{=}100$ & $T{=}200$ \\
    \midrule
    Qwen3.5-4B          & 48 MB  & 4.7 GB  & 9.4 GB  & 57 MB  & 113 MB \\
    Qwen3.5-9B          & 48 MB  & 4.7 GB  & 9.4 GB  & 57 MB  & 113 MB \\
    Qwen3.5-27B         & 144 MB & 14.1 GB & 28.1 GB & 151 MB & 302 MB \\
    Qwen3.5-122B-A10B   & 144 MB & 14.1 GB & 28.1 GB & 142 MB & 283 MB \\
    \bottomrule
  \end{tabular}
\end{table}

Both representations scale linearly with $T$, but their coefficients differ:
\tabref{tab:transient-memory} shows that \sys reduces the per-node transient
footprint by $82$--$99\times$ relative to full-state snapshots.

\runinhead{L3: Additional tree tokens are no longer near-free beyond a hardware-dependent threshold.}
The roofline model~\cite{zhang2025llmcompass, williams2009roofline} explains this latency threshold.  Consider the MLP portion
of target-model execution.  At small row counts, its GEMMs are
weight-bandwidth-bound, so adding verification rows incurs little additional
latency.  Beyond a hardware-dependent threshold, they become compute-bound and
latency grows approximately linearly with row count~\cite{zhao2024atom, li2025orches}.
\figref{fig:motivation}(c) places this transition near 128 rows on A100,
compared with 256 rows on GB10.  This component-level transition illustrates
why additional tree tokens are not indefinitely free.  For a complete hybrid
forward, the efficient capacity further depends on each request's KV length~\cite{kamath2025podattention} and
tree shape~\cite{chen2024sequoia, li2026adaserve}. Fixed per-request or platform-agnostic capacities are
unreliable.

\runinhead{Requirement.}
Together, these limitations call for a tree verifier that is parallel, incurs
low memory overhead, and automatically controls verification cost. Meeting
this requirement raises the three challenges outlined in
\secref{sec:introduction}: tree-coupled recurrent dependencies, state
divergence before branch selection, and coupled, configuration-dependent cost.
These challenges motivate our joint algorithm, kernel, and runtime design,
which we describe in the next section.

\section{The \sys Architecture}
\label{sec:architecture}

\figref{fig:architecture} gives an overview of \sys, a kernel--runtime
co-design for parallel, memory-efficient tree verification in hybrid models.
Its hardware-aware scheduler and target-model execution engine jointly control
verification cost, execute the selected request trees, and commit the accepted
states.

\begin{figure}[t]
  \centering
  \includegraphics[width=\columnwidth]{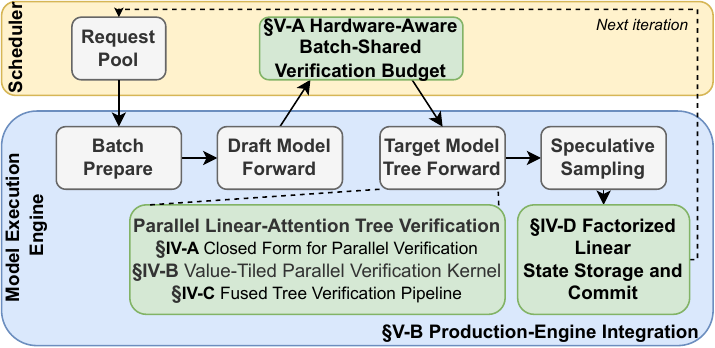}
  \caption{\sys architecture and section roadmap.}
  \label{fig:architecture}
\end{figure}

\runinhead{Overall workflow.}
\sys separates one-time calibration from the recurring serving loop.  Before
serving, a one-time profiler records a hardware-efficient
batch-level verification capacity.  In every target iteration,
\textbf{(1)} the scheduler batches active requests and reads this pre-profiled
capacity.  \textbf{(2)} The drafter generates a candidate tree for each request,
and the scheduler prunes these trees so that their total number of selected
tokens does not exceed the batch budget.  \textbf{(3)} The target verifies the
selected request trees, using tree masks in full-attention layers and \sys's
parallel path in recurrent linear-attention layers.  \textbf{(4)} After
sampling, the engine retains each request's accepted KV entries and
commits its matching linear-attention state.

\runinhead{Design 1: Closed-form formulation for parallel tree
linear-attention verification and a value-tiled GPU kernel (solves L1).}
\secref{sec:closed-form-tree} derives \sys's exactly equivalent all-node closed
form for linear attention over a proposal tree.  \secref{sec:tree-solver} maps
it to a value-tiled GPU kernel that reuses tree factors across value tiles,
keeps finite-Neumann propagation on chip, and fuses output construction.
\secref{sec:fused-pipeline} integrates the kernel with Gated DeltaNet's
convolutional front end and surrounding data transformations.  Together,
these designs provide exact parallel tree
verification while preserving model outputs.

\runinhead{Design 2: Factorized speculative state (solves L2).}
\secref{sec:linear-state-storage} keeps one immutable canonical
linear-attention state and represents candidate branches with token-scale factors.  Once
sampling identifies an accepted path, one batched matrix multiplication
reconstructs its exact state.  Rejected branches use neither state-cache slots
nor rollback.

\runinhead{Design 3: Batch-shared hardware-aware verification budget
(solves L3).}
\secref{sec:hybrid-budget} profiles the complete hybrid target forward for
each execution configuration and records its batch-level capacity.  At runtime,
cumulative draft probability ranks candidates across requests within that
capacity; score monotonicity keeps every request's selected tree
prefix-connected.

\runinhead{Production integration.}
\sys is integrated into SGLang~\cite{zheng2024sglang}, a widely deployed
production LLM serving engine.
\secref{sec:implementation} materializes selected request trees as a
device-resident flat-ragged forest, performs sampling and accepted-path commit
on the GPU, and reuses captured CUDA Graphs~\cite{nvidiaCUDAGraphs} across tree topologies within each
graph capacity.

\section{Parallel Tree Verification for Recurrent Linear Attention}

\subsection{Deriving \sys's Closed Form for Parallel Verification}
\label{sec:closed-form-tree}

\begin{figure}[t]
  \centering
  \includegraphics[width=\columnwidth]{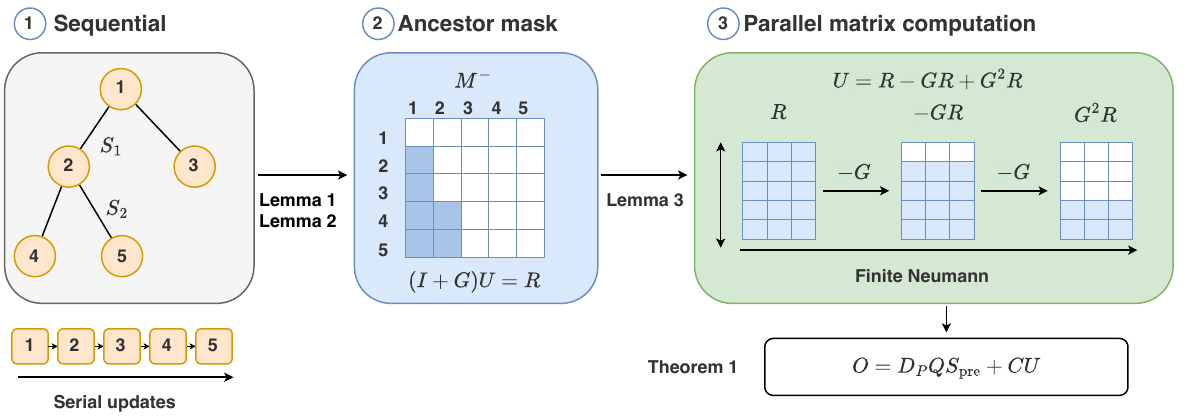}
  \caption{\ra[A.Q2]{From sequential updates to parallel tree verification.}}
  \label{fig:tree-parallel-intuition}
\end{figure}

\ra[A.Q2]{\figref{fig:tree-parallel-intuition} illustrates how \sys exposes parallelism.
\textcircled{\scriptsize 1} Conventional sequential verification updates each
node from its parent's state, so descendants wait for ancestors.
\textbf{Lemma~1} factorizes states into ancestor contributions, which
\textbf{Lemma~2} stacks into \textcircled{\scriptsize 2} an ancestor-masked
matrix system, preserving tree dependencies while replacing separate updates
with joint computation. \textbf{Lemma~3} bounds propagation by tree depth,
enabling \textcircled{\scriptsize 3} finite-Neumann matrix operations over
all node rows in parallel. These steps yield \textbf{Theorem~1}.
\figref{fig:kernel-overview} maps independent value-column tiles to parallel
GPU thread blocks (CTAs). Each block processes all tree nodes for its tile
and keeps intermediate results on chip.}
\tabref{tab:closed-form-notation} collects the notation.

\begin{table}[htp]
  \centering
  \caption{Notation for closed-form tree verification.}
  \label{tab:closed-form-notation}
  \footnotesize
  \setlength{\tabcolsep}{2.5pt}
  \renewcommand{\arraystretch}{1.0}
  \begin{tabularx}{\linewidth}{@{}p{0.25\linewidth}X@{}}
    \toprule
    Symbol & Definition \\
    \midrule
    $\mathcal T,T,\pi(i)$ & Proposal tree, node count, and parent of node $i$.
      $\pi(i)=0$ for a root child and $S_0=\Spre$. \\
    $\Anc(i),d$ & Strict draft ancestors of $i$ and the maximum number of
      strict ancestors of any node. \\
    $H,d_k,d_v,\Spre$ & Number of independent state heads, per-head key/value
      dimensions, and the committed per-head pre-tree state
      $\Spre\in\mathbb R^{d_k\times d_v}$. \\
    $Q,K,V$ & Stacked token operands in
      $\mathbb R^{T\times d_k},\mathbb R^{T\times d_k},
      \mathbb R^{T\times d_v}$, with rows $q_i^\top,k_i^\top,v_i^\top$. \\
    $\beta,\gamma$ & Write and decay gates in $\mathbb R^T$, with
      $\gamma_i>0$. \\
    $\widetilde S_i,u_i,S_i,o_i$ & Sequential rule:
      $\widetilde S_i=\gamma_iS_{\pi(i)}$,
      $u_i=\beta_i(v_i-\widetilde S_i^\top k_i)$,
      $S_i=\widetilde S_i+k_iu_i^\top$, and $o_i=S_i^\top q_i$. \\
    $P_i,P,D_P,D_\beta$ &
      $P_i=\prod_{r\in\Anc(i)\cup\{i\}}\gamma_r$,
      $P=(P_1,\ldots,P_T)^\top$,
      $D_P=\operatorname{diag}(P)$, and
      $D_\beta=\operatorname{diag}(\beta)$. Both diagonal matrices are in
      $\mathbb R^{T\times T}$. \\
    $I,\ind[\cdot],\odot$ & Identity matrix
      $I\in\mathbb R^{T\times T}$, indicator function, and element-wise
      multiplication. \\
    $M^-,M^+$ & Ancestor masks in $\{0,1\}^{T\times T}$:
      $M^-_{ij}=\ind[j\in\Anc(i)]$ and $M^+=M^-+I$. \\
    $G,C,R$ &
      $G,C\in\mathbb R^{T\times T}$ and
      $R\in\mathbb R^{T\times d_v}$, where
      $G=D_\beta D_P[(KK^\top)\odot M^-]D_P^{-1}$,
      $C=D_P[(QK^\top)\odot M^+]D_P^{-1}$, and
      $R=D_\beta(V-D_PK\Spre)$. \\
    $U,O$ & Stacked correction and output rows in
      $\mathbb R^{T\times d_v}$: row $i$ is $u_i^\top$ or $o_i^\top$. \\
    \bottomrule
  \end{tabularx}
\end{table}

\begin{theorem}[\sys's closed-form tree verification]
Sequential tree execution of the gated delta rule produces
\begin{compactequation}
\boxed{
O=D_PQ\Spre+
C(I+G)^{-1}D_\beta\bigl(V-D_PK\Spre\bigr).}
\label{eq:main-closed-form}
\end{compactequation}
This algebraically exact identity replaces sequential traversal and per-node
linear-attention states.
\end{theorem}

Here $D_PQ\Spre$ is the decayed readout of the shared state, $R$ forms
the local deltas, $(I+G)^{-1}$ propagates them along ancestors, and $C$
performs all readouts. Because nodes are indexed parent before child and
$G_{ij}$ can be nonzero only when $j$ is a strict ancestor of $i$, every
nonzero entry lies below the diagonal. Hence $G$ is strictly lower triangular.
Meanwhile, $M^+$ adds the local update required by each output. The identity
therefore exposes a small tree interaction shared by many value-channel
right-hand sides.

\begin{lemma}[Path-factorized state]
For every node $i$, the pre- and post-update states are
\begin{compactequation}
\begin{aligned}
\widetilde S_i
 &=P_i\Spre+\sum_{j\in\Anc(i)}\frac{P_i}{P_j}k_ju_j^{\top},\\
S_i
 &=P_i\Spre+\sum_{j\in\Anc(i)\cup\{i\}}
                 \frac{P_i}{P_j}k_ju_j^{\top}.
\end{aligned}
\label{eq:path-factorization}
\end{compactequation}
\end{lemma}

\begin{proof}[Proof sketch]
Topological induction on $\gamma_iS_{\pi(i)}$ gives the first line with
gate products $P_i/P_j$. Adding the local rank-one update gives the second.
\end{proof}

Reading the state in \textbf{Lemma~1} with $q_i$ gives
\begin{compactequation}
o_i^\top=P_iq_i^\top\Spre+
\sum_{j\in\Anc(i)\cup\{i\}}
\frac{P_i}{P_j}(q_i^\top k_j)u_j^\top.
\label{eq:factorized-readout}
\end{compactequation}
The coefficient of $u_j$ is precisely $C_{ij}$. Stacking
\eqnref{eq:factorized-readout} yields $O=D_PQ\Spre+CU$.

\begin{lemma}[Ancestor-masked correction system]
Let $U\in\mathbb{R}^{T\times d_v}$ stack the correction vectors
$u_i^{\top}$.  Then
\begin{compactequation}
(I+G)U=R.
\label{eq:correction-system}
\end{compactequation}
\end{lemma}

\begin{proof}[Proof sketch]
Substitution into $u_i$ gives $G_{ij}u_j$ for each strict ancestor, with
$G_{ij}=\beta_i(P_i/P_j)k_i^\top k_j$. Stacking yields the system.
The mask excludes siblings, descendants, and the local node, preserving
exactly the sequential dependencies.
\end{proof}

\begin{lemma}[Depth-bounded nilpotency]
With maximum strict-ancestor count $d$, $G$ is depth-bounded nilpotent,
meaning a finite power vanishes under parent-before-child ordering.
\begin{compactequation}
G^{d+1}=0,\qquad
(I+G)^{-1}=\sum_{m=0}^{d}(-G)^m.
\label{eq:nilpotency}
\end{compactequation}
\end{lemma}

\begin{proof}[Proof sketch]
A nonzero entry of $G^m$ represents an $m$-edge ancestor path.
No such path exceeds $d$. It follows that $G^{d+1}=0$, and
multiplication by $I+G$ makes the series telescope.
\end{proof}

\begin{proof}[Proof of Theorem~1]
The unit lower-triangular $I+G$ is invertible. Lemma~2 supplies
$U=(I+G)^{-1}R$, and Lemma~3 makes the inverse finite.
Substitution into the readout following Lemma~1 proves Theorem~1.
\end{proof}
The theorem specifies outputs independently of the solver or hardware
schedule. \sys's value-tiled kernel (\secref{sec:tree-solver}) evaluates
\textbf{Lemma~3}'s finite polynomial, with propagation steps fixed by the
maximum supported depth before execution.

\subsection{Value-Tiled Parallel Verification Kernel}
\label{sec:tree-solver}

\secref{sec:closed-form-tree} transforms the sequential recurrent updates into
an execution-order-independent closed form, but leaves its hardware mapping
unspecified. We realize this form as a value-tiled GPU kernel that shares the
tree-dependent factors $G,C$ across value tiles, assigns each tile to one CTA,
and evaluates the finite polynomial with fixed-shape tensor-core operations
while retaining transient intermediates on chip. We next describe the tile
decomposition, the execution of one CTA, and its operand lifetime.

\runinhead{Value-domain decomposition.}
Mapping all $d_v$ value columns to one CTA per head requires
$\Theta(\widetilde T^2+\widetilde T d_v+d_kd_v)$ on-chip storage and limits
concurrency, where $\widetilde T$ is the fixed extent covering all $T$ nodes.
Since the tree-dependent $G,C$ factors are independent of value channels,
\sys partitions the columns into $N_v=\lceil d_v/b_v\rceil$ sets
$\mathcal J_r$ of at most $b_v$ columns. Let
$X^{(r)}=X[:,\mathcal J_r]$ for
$X\in\{\Spre,V,B_0,R,U,O\}$.  Restricting \eqnref{eq:main-closed-form} to
one value tile gives
\begin{compactequation}
\begin{alignedat}{2}
B_0^{(r)} & =D_P(Q\Spre^{(r)}),\: &
R^{(r)} & =D_\beta\!\left(V^{(r)}-D_P(K\Spre^{(r)})\right),\\
U^{(r)} & =\sum_{m=0}^{d}(-G)^mR^{(r)},\: &
O^{(r)} & =B_0^{(r)}+CU^{(r)}.
\end{alignedat}
\label{eq:kernel-dataflow}
\end{compactequation}
Concatenating the tile results reconstructs $U$ and $O$, so the decomposition
changes only the hardware schedule. Each CTA processes one value tile while
covering the full $\widetilde T$ node extent. The tree is not partitioned
across CTAs.

\begin{figure*}[t]
  \centering
  \includegraphics[width=0.98\textwidth]{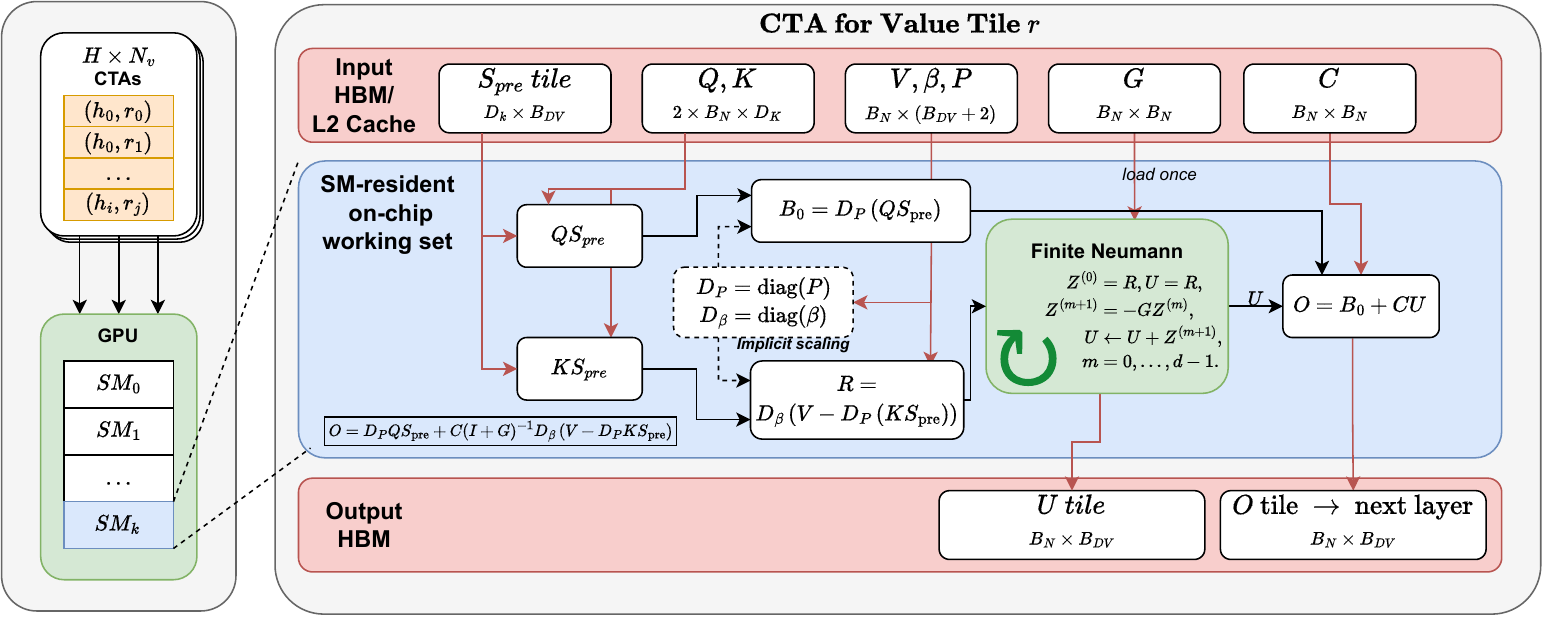}
  \caption{Value-tiled CTA for parallel tree verification
  ($D_K=d_k$, $B_N=\widetilde T$, and $B_{DV}=b_v$).}
  \label{fig:kernel-overview}
\end{figure*}

\runinhead{One-CTA execution.}
Before the value-tile CTAs run, a factor producer forms $KK^\top$ and
$QK^\top$, then applies the ancestor masks, decay ratios, and $\beta$ to
construct $G$ and $C$.  It writes them once to a layer-local workspace,
amortizing the $O(\widetilde T^2d_k)$ Gram construction across all $N_v$ CTAs.

As \figref{fig:kernel-overview} shows, one CTA reads
$\Spre^{(r)}\in\mathbb R^{d_k\times b_v}$ and
$V^{(r)}\in\mathbb R^{\widetilde T\times b_v}$, then forms
$Q\Spre^{(r)}$ and $K\Spre^{(r)}$.  These products are traversed along
$d_k$ in chunks of $b_k$ and decomposed into tensor-core matrix
multiply-accumulate (MMA) fragments.  Although
\eqnref{eq:kernel-dataflow} uses $D_P$ and $D_\beta$, the implementation
stores only the length-$T$ vectors $P$ and $\beta$ and applies them row-wise to
register tiles.  It never materializes either $T\times T$ diagonal matrix.

\runinhead{On-chip finite-Neumann computation.}
The CTA next initializes $Z^{(0)}=R^{(r)}$ and $U^{(r)}=R^{(r)}$, then executes
the following fixed sequence on-chip
\begingroup
\small
\begin{compactequation}
Z^{(m+1)}=-GZ^{(m)},\quad
U^{(r)}\mathrel{+}=Z^{(m+1)},\quad 0\le m<d.
\label{eq:neumann-kernel}
\end{compactequation}
\endgroup
Lemma~3 guarantees $U^{(r)}=(I+G)^{-1}R^{(r)}$. Node-dimension MMA
fragments process all tree positions in parallel, updating $Z$ in place
and accumulating into $U^{(r)}$. All $d$ rounds remain on chip in one CTA,
which then micro-tiles $CU^{(r)}$ to form $O^{(r)}=B_0^{(r)}+CU^{(r)}$.
This finite polynomial is algebraically exact. \cq[CQ4]{\sys matches SGLang-Tree's BF16 input/output and FP32
recurrent-state and accumulation precision, with FP32 intermediate factors.}

\runinhead{Operand lifetime and on-chip reuse.}
The CTA releases $\Spre^{(r)}$ after the two state products and no longer needs
$P$ or $\beta$ after forming $B_0^{(r)}$ and $R^{(r)}$.  $G$ remains live only
during propagation, while $C$ is loaded for the final readout.  The register
tiles for $B_0^{(r)}$ and $U^{(r)}$ survive until output construction.
$B_0^{(r)}$, $R^{(r)}$, and $Z^{(m)}$ are never written to HBM.  Only the final
$O^{(r)}$ tile and the $U^{(r)}$ tile required by state commit leave the CTA.

\runinhead{Two-dimensional resource tiling.}
The value width $b_v$ balances CTA concurrency against reuse. Wider tiles
amortize factor reads and setup but enlarge register and shared-memory
footprints. Narrower tiles expose more CTAs but duplicate those costs.
The key dimension is traversed in $b_k$-wide chunks with a single-stage
pipeline, retaining only the current $Q/K$ chunk and matching
$\Spre^{(r)}$ slice. Wider chunks reduce partial accumulations at the cost
of more transient storage. Together, the widths bound the CTA working set by
$M_{\rm CTA}=\Theta\!\left(
\widetilde T^2+\widetilde T b_v+b_k(\widetilde T+b_v)\right)$,
and expose $HN_v$ CTAs per request. \sys selects $(b_k,b_v)$ from
compile-time configurations to balance resource use, concurrency, and reuse.

\subsection{Fused Tree-Verification Pipeline}
\label{sec:fused-pipeline}

\secrefs{sec:closed-form-tree}{sec:tree-solver} parallelize the
recurrent core in \eqnref{eq:gdn-background}. A complete Gated DeltaNet
layer also executes the short causal convolution introduced in
\secref{sec:background-hybrid} and surrounding gate and layout operations.
\sys organizes the complete tree execution into one round-level topology step
followed by three per-layer GPU stages:
(1) tree-parallel convolution and gate preparation, (2) factor construction,
and (3) value-tiled solve and readout.

\runinhead{Round-level topology preprocessing.}
The parent map, $M^-$, $M^+$, and maximum depth depend only on the tree. \sys
builds them once per request and round and reuses them across layers and heads.
Fixed buffers avoid repeated transitive closure and support captured graphs.

\runinhead{Stage 1: Tree-parallel convolution and gate preparation.}
On a linear sequence, the convolution reads a fixed number of preceding
tokens. For a proposal-tree node, these inputs lie on its root-to-node path,
not at adjacent positions in the flattened tree. Each
$(\text{node},\text{channel tile})$ program follows parent links and fuses
convolution, bias, SiLU, and gate preparation. All nodes execute concurrently;
the fixed kernel width makes per-node work independent of proposal-tree depth.

\runinhead{Stage 2: Factor construction.}
One producer constructs $G,C$ per layer and head for reuse by all value tiles.

\runinhead{Stage 3: Value-tiled solve and readout.}
Value-tiled CTAs produce $O$ and compact $U$ commit factors. Normalization,
path decay, and layout conversion are folded into these stages. Request
indices select canonical head-major states. The factor/solve boundary
avoids redundant construction of $G,C$ across tiles.

\subsection{Factorized Linear-State Storage and Commit}
\label{sec:linear-state-storage}

\figref{fig:state-lifetime} follows the state data across verification,
sampling, and commit. Unlike an append-only KV
cache, a linear-attention state is one large
mutable summary.  Forking it for every proposal node would require a
full $d_k\times d_v$ snapshot per head and branch.  \sys instead keeps
one committed state per active request and layer in an HBM slot. This
$\Spre$ remains immutable throughout verification, so no branch owns a
private state or requires rollback.  Slots use the serving runtime's
canonical $[H,d_v,d_k]$ layout. Verification reads the existing slot,
and commit overwrites it only after the accepted path is known. The
verification kernel's $U$ output directly  feeds the state lifecycle shown in
\figref{fig:state-lifetime}.

\begin{figure}[t]
  \centering
  \includegraphics[width=0.99\columnwidth]{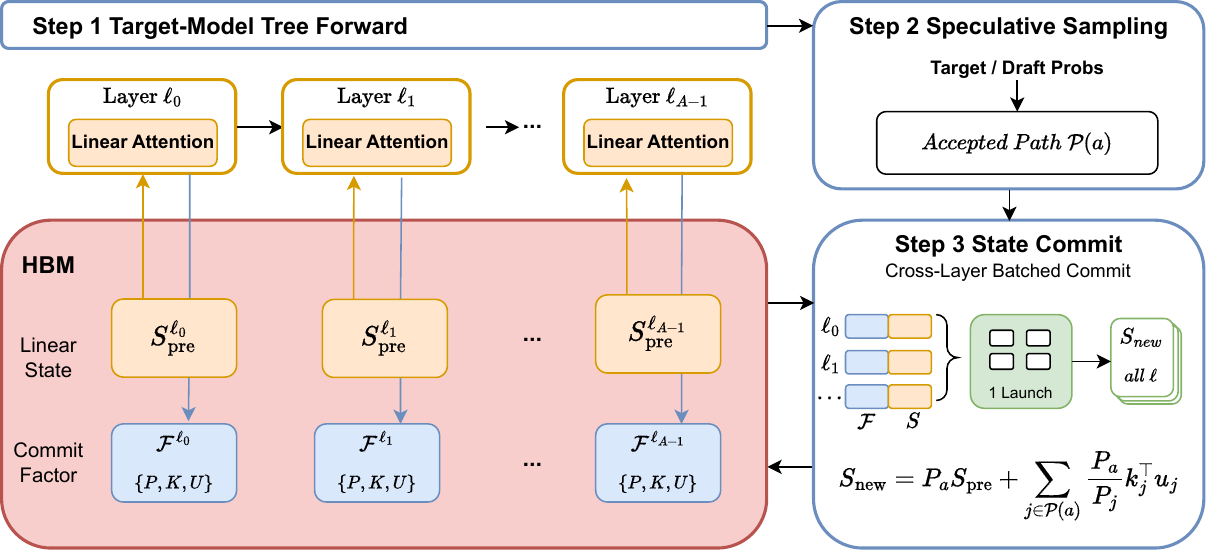}
  \caption{Factorized speculative-state lifecycle.}
  \label{fig:state-lifetime}
\end{figure}

\runinhead{Factorized speculative states.}
Materializing $T$ candidate states has space complexity
$\Theta(THd_kd_v)$. By Lemma~1, \sys retains
only the node factors $P,K,U$ in addition to the single committed
state, reducing the speculative-state storage complexity to
\begin{compactequation}
\underbrace{\Theta(Hd_kd_v)}_{\text{committed state}}
+\underbrace{\Theta\!\left(TH(d_k+d_v+1)\right)}_{\text{tree factors}}.
\label{eq:factorized-storage}
\end{compactequation}
For grouped key heads, each shared $K$ factor is counted once.
When $d_k=d_v=d_s$, a full state costs $\Theta(Hd_s^2)$ elements per
branch, whereas one factorized node costs only $\Theta(2Hd_s)$.  The
advantage therefore grows with the state dimension rather than relying
on a particular proposal-tree shape.

Canonical states persist with requests, $P,K,U$ until acceptance, and
$G,C$ and Neumann workspaces only within a layer. Scratch is reused across
layers and rounds. Rejected branches consume no state-cache slots, making
state-cache capacity independent of draft-node count.

\runinhead{Accepted-path commit.}
For one head, suppose verification selects terminal node $a$.  Let
$\mathcal P(a)=\Anc(a)\cup\{a\}$ be its accepted path, let
$K_a\in\mathbb R^{|\mathcal P(a)|\times d_k}$ and
$U_a\in\mathbb R^{|\mathcal P(a)|\times d_v}$ stack the path factors in
topological order, and define
$D_a=\operatorname{diag}(P_a/P_j)_{j\in\mathcal P(a)}$.
Lemma~1 reconstructs the only state that must persist:
\begin{compactequation}
S_{\rm new}=P_a\Spre+(D_a K_a)^\top U_a.
\label{eq:accepted-commit}
\end{compactequation}
\sys compacts the accepted path once per round and evaluates the second
term as one tensor-core matrix multiplication.  The selected path is
shared across all linear layers, allowing their commits to be folded
into one device launch that writes each layer's canonical cache slot.
The old state remains readable until this launch completes, so no
additional state snapshot or copy-on-write slot is needed.  If no draft
token is accepted, the slot retains $\Spre$ unchanged.

\section{System Optimization}

Beyond parallel tree execution, serving must decide how much work to admit
across concurrent requests and map the resulting irregular forest to a
graph-captured GPU iteration.  \sys combines hardware-aware batch-level
scheduling with a device-resident production runtime.

\subsection{Hardware-Aware Batch-Level Verification Scheduling}
\label{sec:hybrid-budget}

\runinhead{Separating capacity from proposal utility.}
Given one candidate tree per request, the scheduler answers two distinct
questions: how many nodes the hybrid target can verify in its hardware-efficient
regime, and which nodes should consume that capacity.  A node count alone is
not a stable cost model.  For request $i$, let $n_i$, $L_i$, and $d_i$ denote
its selected node count, committed KV length, and tree depth; let
$N=\sum_i n_i$.  The target latency can be viewed as
\begin{compactequation}
\begin{aligned}
T_{\rm ver}(\mathbf n,\mathbf L,\mathbf d)={}&
T_{\rm full}(\mathbf n,\mathbf L)
+T_{\rm linear}(\mathbf n,\mathbf d)\\
&+T_{\rm dense}(N)+T_{\rm other}(N).
\end{aligned}
\label{eq:hybrid-cost-decomposition}
\end{compactequation}
Full-attention layers read the committed prefixes, so their dominant KV work
depends on both $n_i$ and $L_i$.  Recurrent layers avoid this prefix scan, but
their tree interaction and propagation schedules depend on $n_i$ and $d_i$.
The shared projections and MLPs batch all $N$ nodes and cross from
weight-bandwidth-bound to compute-bound execution at a platform-specific
operating point. \ra[A.Q3]{Our contribution is a hybrid-attention verification
budget that calibrates the combined cost of full-attention, recurrent, and
dense operators to determine total admitted work across requests.}

\runinhead{\ra{Startup calibration of verification capacity.}}
\ra[A.Q3]{The serving engine's standard startup procedure determines the model,
GPU, parallelism, and execution settings. \sys uses these settings to calibrate
during warmup. For each supported runtime
configuration $c$, it measures complete target-forward latency
$T_{\rm ver}(N\mid c)$ over the engine's default CUDA Graph capacity set
$\mathcal G$ and ordinary one-token-per-request decode latency $T_{\rm dec}(c)$.
Both include recurrent propagation, full attention, MLPs, and communication
under the deployed kernels and CUDA Graphs. \sys selects the batch-wide capacity}
\begingroup
\small
\begin{compactequation}
B_{\rm ver}(c)=\max\left\{N\in\mathcal G:\
T_{\rm ver}(N\mid c)\le(1+\epsilon)T_{\rm dec}(c)\right\},
\label{eq:profiled-budget}
\end{compactequation}
\endgroup
\ra{where $\epsilon$ bounds the additional target-forward latency relative
to ordinary decoding.} \re{We set the default allowance to $\epsilon=0.05$.}
\ra[A.Q3]{Calibration adds at most 2.9\% to engine startup
time. At runtime, batching metadata selects the calibrated batch-wide
capacity $B_{\rm ver}(c)$.}

\runinhead{Batch-wide utility allocation.}
The calibrated capacity fixes the total verification work; \sys then distributes
it according to the drafter's estimate of node utility.  For node $v$, let
$\mathcal P(v)$ be its root-to-$v$ path and $\pi(u)$ the parent of node $u$.
We score node $v$ as
$\rho(v)=\prod_{u\in\mathcal P(v)}
p_{\rm draft}\bigl(u\mid\pi(u)\bigr)$,
the probability assigned by the drafter to reaching $v$ along its complete
path.  After retaining the root of each active tree, a global selection fills
the remaining capacity with the highest-scoring nodes across the batch.
Because a descendant's cumulative probability cannot exceed its parent's, the
selected nodes of each request are prefix-connected without a separate repair
pass.  A segmented count of the selected nodes produces a variable budget
$q_i$ for each request while preserving the fixed batch total.  This allocation
maximizes the retained cumulative probability mass under the calibrated
capacity, shifting work toward the most promising continuations.  The selector
is implemented with fixed-shape device top-$k$ and scatter operations, allowing
allocation to remain inside the captured GPU iteration.

\subsection{Production-Engine Integration}
\label{sec:implementation}

We integrate \sys into SGLang, one of the most widely deployed production LLM
inference engines, with approximately 6.2\,kLoC of Python and Triton~\cite{tillet2019triton} code.
The implementation comprises parallel verification kernels that can be readily
ported to other inference engines, together with deep optimizations to SGLang's
speculative-decoding and execution runtime.

\runinhead{CUDA Graph-native packed-forest execution.}
The integration connects \sys's verifier to SGLang's GDN backend, EAGLE loop,
CUDA Graph runners, and linear-state pool. \rd[D.Q2]{The serving engine already configures the
maximum proposal depth at startup. \sys JIT-compiles its kernel for this depth,
accommodating runtime tree variation within it.}
The allocator's selected node IDs
and per-request budgets feed a device-side materializer, which topologically
packs the variable-sized trees into one flat-ragged forest.  Request offsets,
parent indices, and compact ancestor metadata describe the forest without
changing its total capacity.  A request-aware block map then emits only
verification tiles internal to each request, so compacting the storage also
reduces the work presented to the recurrent kernels.  The same packed topology
drives full-attention layers through SGLang's tree-attention path and recurrent
layers through \sys's parallel verifier.  Consequently, fixed-capacity CUDA
Graphs retain fixed addresses and kernel shapes while the node allocation and
tree topology vary from round to round.

\runinhead{Unified cross-layer state commit.}
Each recurrent layer writes its compact $P,K,U$ factors directly into a fixed
region of a preallocated contiguous GPU buffer.  Target sampling produces one accepted
path descriptor per request, and device-side compaction maps these paths to the
corresponding packed rows.  A single batched launch consumes the descriptors
across recurrent layers to evaluate \eqnref{eq:accepted-commit}, while
SGLang commits the matching KV entries for full-attention layers.  The buffer is
reused by the next round, providing a common data path from verification to
both hybrid-model state representations.

\runinhead{Bubble-free GPU execution.}
Selection, forest materialization, target execution, sampling, path compaction,
and KV/GDN state commit remain on the GPU as one captured execution flow.  This
eliminates host-induced synchronization gaps and allows the CPU scheduler to
prepare subsequent work while the current target iteration executes.  The GPU
pipeline therefore remains continuously supplied across speculative rounds,
sustaining high utilization.

\begin{figure*}[!t]
  \centering
  \includegraphics[width=0.98\textwidth]{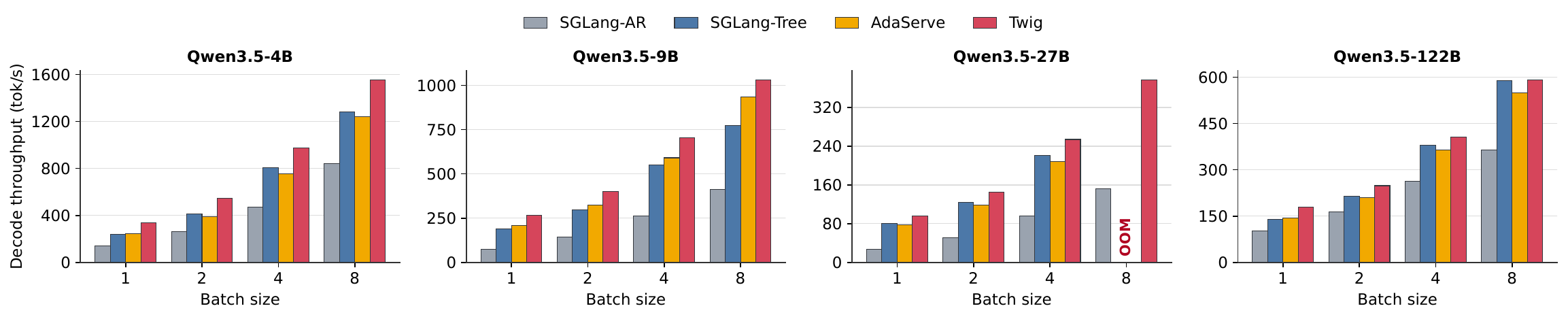}
  \caption{A100 MBPP throughput across Qwen3.5 models and batch sizes.}
  \label{fig:eval-offline-a100}
  \vspace{-12pt}
\end{figure*}

\begin{figure*}[!t]
  \centering
  \begin{minipage}[t]{0.71\textwidth}
    \vspace{0pt}
    \centering
    \includegraphics[width=\linewidth]{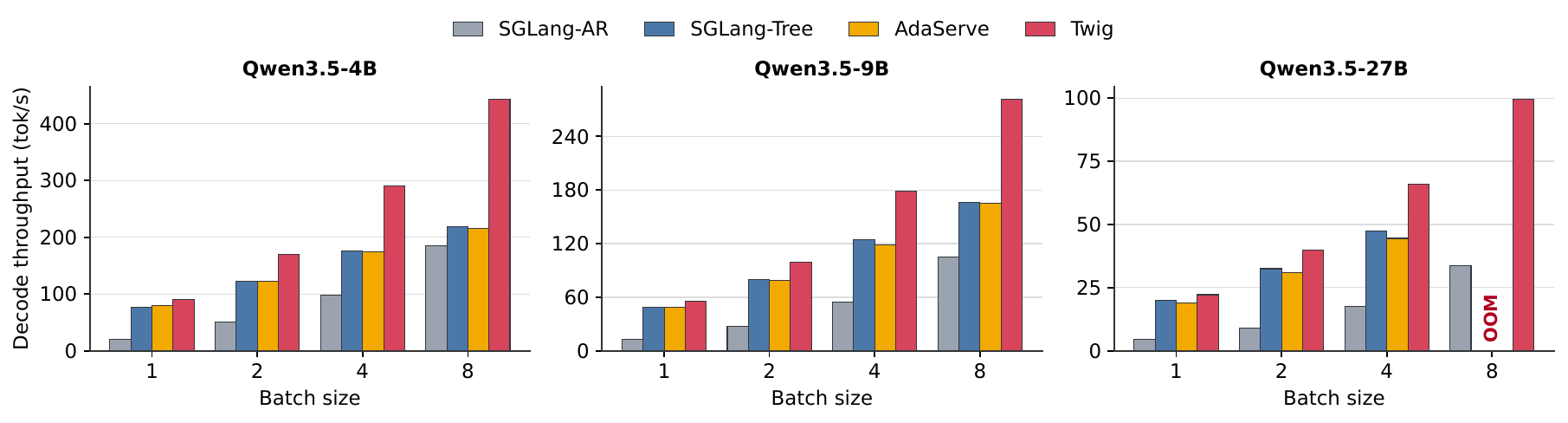}
    \caption{GB10 MBPP throughput across models and batch sizes. OOM indicates
    out of memory.}
    \label{fig:eval-offline-gb10}
  \end{minipage}
  \hfill
  \begin{minipage}[t]{0.27\textwidth}
    \vspace{0pt}
    \centering
    \makeatletter
    \def\@captype{table}
    \makeatother
    \caption{Mean accepted tokens (MAT) on MBPP.}
    \label{tab:eval-offline-mat}
    \scriptsize
    \setlength{\tabcolsep}{1.5pt}
    \renewcommand{\arraystretch}{1.18}
    \begin{tabular*}{\linewidth}{@{\extracolsep{\fill}}lrrr@{}}
      \toprule
      Model & \shortstack{SGLang-\\Tree} & AdaServe & \sys \\
      \midrule
      Qwen3.5-4B       & 6.29 & 6.35 & 6.38 \\
      Qwen3.5-9B       & 6.76 & 6.82 & 6.88 \\
      Qwen3.5-27B      & 6.56 & 6.61 & 6.62 \\
      Qwen3.5-122B-A10B & 6.31 & 6.39 & 6.40 \\
      \bottomrule
    \end{tabular*}
  \end{minipage}
  \vspace{-12pt}
\end{figure*}

\section{Evaluation}
\label{sec:evaluation}

After describing the experimental setup (\secref{sec:eval-setup}), we evaluate
end-to-end decode throughput (\secref{sec:eval-offline}), generality across
workloads (\secref{sec:eval-datasets}), \cq{generality across models
(\secref{sec:eval-models}),} and real-world online agent serving
(\secref{sec:eval-online}).  We then isolate parallel-verification kernel
efficiency (\secref{sec:eval-kernel}), validate the hardware-aware verification
budget through sensitivity analysis (\secref{sec:eval-budget}),
\cq{validate numerical accuracy (\secref{sec:eval-accuracy}),} and quantify
each design's contribution through an ablation study
(\secref{sec:eval-ablation}).

\subsection{Evaluation Setup}
\label{sec:eval-setup}

\runinhead{Platforms and models.}
We use two platforms with distinct compute-to-memory ratios:
\begin{itemize}[leftmargin=1.3em,nosep]
  \item \emph{A100.} One server with four NVLink-connected NVIDIA A100
  80\,GB SXM4 GPUs, each with 312 FP16/BF16 TFLOP/s and 2.04\,TB/s
  HBM2e bandwidth.
  \item \emph{GB10.} One NVIDIA DGX Spark system powered by the GB10
  Grace Blackwell Superchip, with 256 FP16/BF16 TFLOP/s and 128\,GB of
  coherent LPDDR5x at 273\,GB/s.
\end{itemize}
\tabref{tab:eval-deployment} maps four Qwen3.5~\cite{qwen2026qwen35}
models, from 4B to 122B parameters, to these platforms.  We use each
Qwen3.5 model's native multi-token-prediction (MTP) head as the drafter, execute
unquantized weights, and use NVLink-connected tensor parallelism
(TP)~\cite{Narayanan2021megatron} only for Qwen3.5-122B-A10B.

\begin{table}[t]
  \centering
  \caption{\ra[A.Q4]{Model deployment with MTP parameters and weight memory.}}
  \label{tab:eval-deployment}
  \footnotesize
  \setlength{\tabcolsep}{3pt}
  \begin{tabular*}{\linewidth}{@{\extracolsep{\fill}}lrcc@{}}
    \toprule
    Model & \ra{\shortstack{MTP\\params / GB}}
          & \shortstack{A100\\4$\times$80\,GB}
          & \shortstack{GB10\\128\,GB} \\
    \midrule
    Qwen3.5-4B       & \ra{120M / 0.241} & TP1 & TP1 \\
    Qwen3.5-9B       & \ra{243M / 0.487} & TP1 & TP1 \\
    Qwen3.5-27B      & \ra{425M / 0.849} & TP1 & TP1 \\
    Qwen3.5-122B-A10B & \ra{2.52B / 5.05} & TP4 & -- \\
    \bottomrule
  \end{tabular*}
\end{table}

\runinhead{Baselines.}
We compare against three baselines in SGLang~\cite{GitHubsglang}, one of the
most widely deployed production LLM serving engines.  \emph{SGLang-AR} performs
standard autoregressive decoding.
\emph{SGLang-Tree} uses SGLang's native tree-speculative decoding path.
\emph{AdaServe}~\cite{li2026adaserve} is a state-of-the-art
tree-speculative serving system.  Because its open-source implementation does
not support the evaluated hybrid-attention models, we port and optimize it
in SGLang.  For Qwen3.5, all tree methods use the same native MTP drafter, model weights,
sampling configuration, and total proposal-token budget,
\re[E.Q2]{while retaining their respective node-selection and pruning policies.} We set
top-$k=4$ and the maximum depth to 8.

\runinhead{Workloads.}
Offline experiments span code (MBPP), dialogue (ShareGPT), mathematical
reasoning (GSM8K), and summarization
(CNN/DailyMail)~\cite{austin2021mbpp,sharegpt,cobbe2021gsm8k,see2017cnndailymail}.
For online experiments, we replay a real-world coding-agent workload:
OpenHands~\cite{wang2025openhands} sessions from NVIDIA
Open-SWE-Traces~\cite{ahmad2026openswetraces}.

\runinhead{Metrics and methodology.}
For offline serving, we hold the active batch size fixed and report decode
throughput in accepted output tokens per second.  \rf[F.Q4]{For online serving, we
report time to first token (TTFT), including queueing and prefill, and time per
output token (TPOT) after the first token. All measurements use repeated runs.}

\subsection{End-to-End Decode Throughput}
\label{sec:eval-offline}

We evaluate end-to-end offline decode throughput on MBPP using all four
systems.  \figsref{fig:eval-offline-a100}{fig:eval-offline-gb10}
cover all 28 configurations.

\begin{figure*}[!t]
  \centering
  \begin{minipage}[t]{0.71\textwidth}
    \vspace{0pt}
    \centering
    \includegraphics[width=\linewidth]{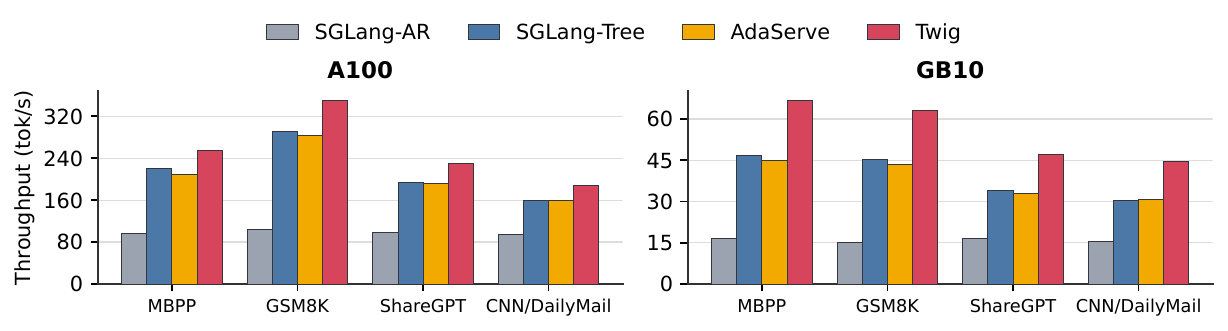}
    \caption{Qwen3.5-27B throughput across workloads on A100 and GB10 at batch
    size 4.}
    \label{fig:eval-datasets}
  \end{minipage}
  \hfill
  \begin{minipage}[t]{0.27\textwidth}
    \vspace{0pt}
    \centering
    \makeatletter
    \def\@captype{table}
    \makeatother
    \caption{MAT across generation workloads.}
    \label{tab:eval-dataset-mat}
    \scriptsize
    \setlength{\tabcolsep}{1.5pt}
    \renewcommand{\arraystretch}{1.18}
    \begin{tabular*}{\linewidth}{@{\extracolsep{\fill}}lrrr@{}}
      \toprule
      Dataset & \shortstack{SGLang-\\Tree} & AdaServe & \sys \\
      \midrule
      MBPP          & 6.57 & 6.68 & 6.70 \\
      GSM8K         & 7.07 & 7.08 & 7.08 \\
      ShareGPT      & 5.31 & 5.50 & 5.51 \\
      CNN/DailyMail & 4.71 & 4.95 & 4.98 \\
      \bottomrule
    \end{tabular*}
  \end{minipage}
  \vspace{-12pt}
\end{figure*}

\begin{figure*}[!t]
  \centering
  \includegraphics[width=0.98\textwidth]{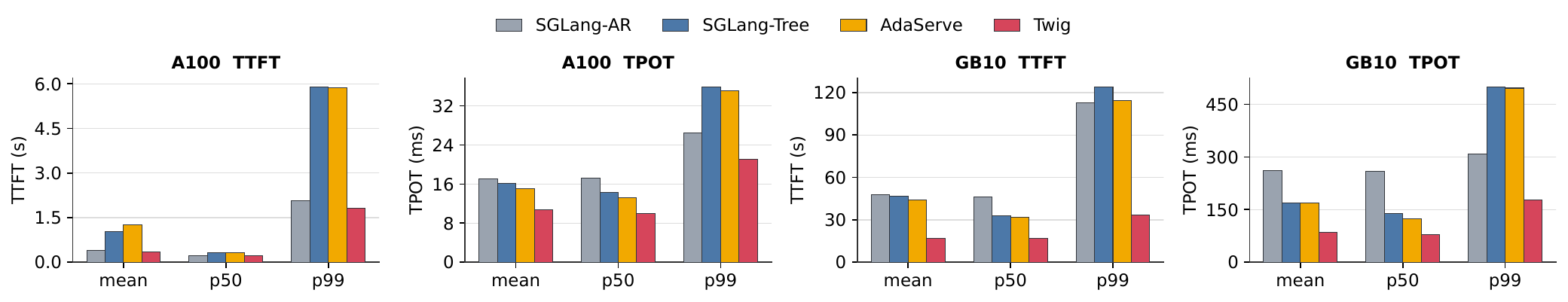}
  \caption{Mean, p50, and p99 online latency on Open-SWE-Traces, using
  Qwen3.5-122B-A10B on A100 and Qwen3.5-27B on GB10. Lower is better.}
  \label{fig:eval-online}
  \vspace{-12pt}
\end{figure*}

\sys consistently improves decode throughput.
Across all configurations, it achieves a
geometric-mean speedup of $2.74\times$ over SGLang-AR and $1.26\times$
over the strongest speculative baseline.
On A100, its peak speedups over SGLang-AR, SGLang-Tree, and AdaServe are
$3.62\times$, $1.41\times$, and $1.39\times$, respectively; on GB10, the
corresponding peak speedups reach $4.72\times$, $2.03\times$, and
$2.06\times$.  \sys delivers the highest throughput across the single-GPU
deployments and in the 122B TP4 deployment, demonstrating that its benefits
extend to tensor-parallel execution.  \tabref{tab:eval-offline-mat} further shows that
our pruning method consistently attains higher MAT than both
speculative baselines across all four models.

The parallel verification kernel substantially accelerates tree verification.
On platforms with a higher compute-to-memory ratio, memory-bound decoding
leaves more compute headroom for a larger proposal budget.  At these larger
budgets, \sys's parallel verifier delivers even greater speedups by efficiently
absorbing the increased verification workload.  On GB10, its
speedup over the strongest speculative baseline increases from
$1.11$--$1.14\times$ at batch size 1 to $1.70$--$2.03\times$ at batch size 8.

\begin{table}[t]
  \centering
  \caption{Measured non-cache speculation overhead on A100 at batch size 4
  with 128 proposal nodes (GiB). Lower is better.}
  \label{tab:eval-memory}
  \footnotesize
  \setlength{\tabcolsep}{4pt}
  \begin{tabular}{@{}lrrr@{}}
    \toprule
    Model & \shortstack{SGLang-\\Tree} & AdaServe & \sys \\
    \midrule
    Qwen3.5-4B        &  8.72 &  7.79 & \textbf{1.53} \\
    Qwen3.5-9B        &  8.74 &  7.81 & \textbf{1.55} \\
    Qwen3.5-27B       & 24.00 & 21.00 & \textbf{2.42} \\
    Qwen3.5-122B-A10B & 28.04 & 25.04 & \textbf{6.52} \\
    \bottomrule
  \end{tabular}
\end{table}

Factorized state management expands the feasible batching capacity while
preserving exact, lossless verification semantics.  \tabref{tab:eval-memory}
reports the additional memory used by tree speculation after subtracting AR
and MTP-drafter memory.  The results match the transient-state analysis in
\tabref{tab:transient-memory}: \sys reduces this overhead by 74.0--88.5\%
relative to AdaServe, freeing 6.26--18.58~GiB for serving caches.  The 122B row
aggregates the four GPUs used for TP4.  This additional capacity allows
\sys to run Qwen3.5-27B at batch size 8 without OOM on either platform, whereas both speculative baselines OOM.

\subsection{Generality Across Workloads}
\label{sec:eval-datasets}

Following the setup in \secref{sec:eval-offline}, we evaluate the other
three workloads at batch size 4.

The workloads vary substantially in prediction difficulty, yielding MATs
from 4.98 on CNN/DailyMail to 7.08 on GSM8K
(\figref{fig:eval-datasets} and \tabref{tab:eval-dataset-mat}).  Relative to
SGLang-AR, \sys achieves
$1.98\times$--$3.36\times$ speedups on A100 and $2.85\times$--$4.12\times$
on GB10.  Relative to the faster speculative baseline, the corresponding
speedups are $1.15\times$--$1.20\times$ and $1.38\times$--$1.45\times$.
Overall, \sys accelerates decoding across diverse workloads.

\subsection{\cq{Generality Across Models}}
\label{sec:eval-models}

\begin{figure}[t]
  \centering
  \captionsetup[subfloat]{captionskip=2pt}
  \subfloat[\cq{Decode throughput}\label{fig:eval-model-throughput}]{%
    \includegraphics[width=0.49\columnwidth]{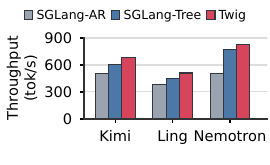}}
  \hfill
  \subfloat[\cq{Kernel execution time}\label{fig:eval-model-kernel}]{%
    \includegraphics[width=0.49\columnwidth]{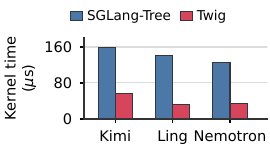}}
  \caption{\cq[CQ2]{Hybrid-model evaluation.}}
  \label{fig:eval-model-generality}
  \vspace{-12pt}
\end{figure}

\cq[CQ2]{We evaluate \sys beyond Qwen3.5 on Kimi Linear~\cite{kimi2025linear},
Ling-3.0-flash~\cite{inclusionai2026ling}, and Nemotron-3.5-Lightning~\cite{nvidia2026lightning}.
The first two use Kimi Delta Attention (KDA), and Nemotron uses Mamba-2,
covering major hybrid-LLM linear recurrences alongside Qwen3.5's Gated DeltaNet.
Each variant needs approximately ten lines of operator-specific core-kernel
adaptation. Kimi Linear uses the block-diffusion drafter DFlash~\cite{chen2026dflash},
while Ling and Nemotron use native MTP heads.}

\cq[CQ2]{In \figref{fig:eval-model-throughput}, \sys improves decode throughput
over SGLang-Tree by 12.3\%, 13.1\%, and 7.3\% on Kimi Linear, Ling, and
Nemotron, respectively. Parallel matrix operations replace node-serial state
updates, accelerating recurrent verification by $2.84\times$, $4.40\times$,
and $3.77\times$ (\figref{fig:eval-model-kernel}).}

\subsection{Real-World Online Agent Serving}
\label{sec:eval-online}

We replay multi-turn Open-SWE-Traces sessions using the largest model
supported by each platform: Qwen3.5-122B-A10B on A100 and Qwen3.5-27B on
GB10.
Following standard LLM-serving methodology~\cite{amey2024sarathiserve},
session start times follow a Poisson process, and every system receives the
same arrival trace on each platform.  \figref{fig:eval-online} reports
mean, p50, and p99 TTFT and TPOT.

\sys reduces mean TTFT and TPOT against every baseline.  Compared with
SGLang-AR, it lowers mean TTFT by 15.8\%--64.3\% and mean TPOT by
37.6\%--67.6\% across the two platforms.  Against SGLang-Tree, the
corresponding reductions are 63.5\%--67.6\% and 33.8\%--50.3\%, while
against AdaServe they are 61.3\%--73.3\% and 28.9\%--49.9\%.

\sys's factorized state representation leaves substantially more memory for
the KV cache, reducing repeated prefills and thus improving TTFT.
Each agent turn extends a long prefix shared with its preceding turns, so
evicting a cached prefix forces a repeated prefill and directly increases
TTFT.  Existing tree systems reserve a recurrent-state snapshot for every
proposal, incurring 82--99$\times$ the transient-memory footprint of
\sys (\tabref{tab:transient-memory}) and leaving less capacity for reusable
prefixes.  By retaining only factorized updates, \sys attains 90.8\% and 92.6\% prefix-cache hit
rates on A100 and GB10.  SGLang-Tree reaches 69.7\% and 58.3\%, while
AdaServe reaches 69.3\% and 59.8\%.  \sys thus recovers 21.1 to 34.3
percentage points of cache hits and nearly matches SGLang-AR's 89.8\% and
93.3\%, despite verifying a proposal tree.

Under continuous batching, TTFT and TPOT interact through the shared
execution stream, and \sys improves both.  TTFT consists of queueing and
prefill: compared with SGLang-AR, \sys's lower TPOT releases batch slots
sooner and shortens queueing, while its factorized state preserves a
comparable prefix-cache hit rate and avoids additional prefills.  Conversely,
fewer prefills also reduce interference with active decoding, reinforcing
the TPOT improvement.

\begin{figure}[t]
  \centering
  \captionsetup[subfloat]{captionskip=2pt}
  \subfloat[\cq{TTFT and TPOT}\label{fig:eval-online-load-latency}]{%
    \includegraphics[width=0.49\columnwidth]{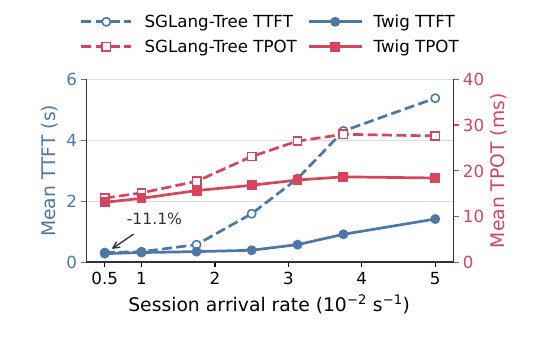}}
  \hfill
  \subfloat[\cq{Prefix-cache hit rate}\label{fig:eval-online-load-cache}]{%
    \includegraphics[width=0.49\columnwidth]{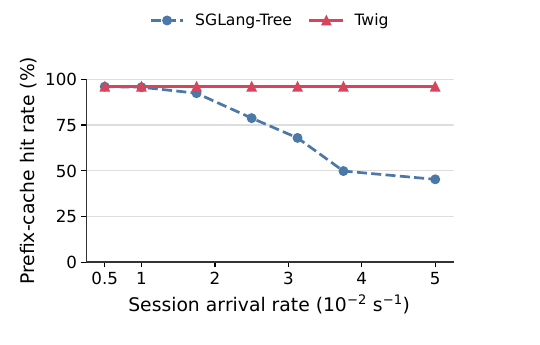}}
  \caption{\cq[CQ5]{Online load sensitivity on Qwen3.5-122B-A10B/A100.}}
  \label{fig:eval-online-load}
  \vspace{-12pt}
\end{figure}

\cq[CQ5]{\figref{fig:eval-online-load} sweeps Poisson session arrival rates
from 0.005 to 0.05~s$^{-1}$ on four A100 GPUs using eight fixed
40-turn sessions averaging 19,169 input and 134 output tokens.
Turns execute sequentially within sessions, with at most four concurrent
requests. Different state footprints leave 79,872 KV-cache tokens for
SGLang-Tree and 262,144 for \sys. At the lowest load with lower cache pressure,
both achieve 96.01\% cache hits, yet \sys reduces mean TTFT/TPOT by
11.1\%/6.5\%. As load increases, \sys maintains 96.01\% hits while
Tree falls to 45.32\%, with mean TTFT/TPOT reductions reaching 79.0\%/33.3\%.
At the highest load, p50 and p99 TTFT/TPOT fall by 74.8\%/12.1\% and
67.1\%/64.8\%, respectively. Across the seven rates, mean and p99 full-request
latency, including prefill, decrease by 34.6\% and 33.3\% over SGLang-Tree.
With the joint SLO of TTFT $\leq2$~s and TPOT $\leq30$~ms, highest-load
attainment is 75.31\% versus Tree's 35.31\%, yielding $4.08\times$ goodput
in SLO-compliant requests per second.}

\subsection{Parallel Verification Kernel Efficiency}
\label{sec:eval-kernel}

To isolate parallel recurrent verification, we compare the value-tiled
verifier from \secref{sec:tree-solver} with SGLang-Tree's serial
delta-rule verifier on Qwen3.5-9B under the same proposal tree and drafter
output. \tabref{tab:kernel-verify} reports complete-core latency and
solve-kernel metrics, both collected using
\texttt{nsys}~\cite{nvidiaNsightSystems} and
\texttt{ncu}~\cite{nvidiaNsightCompute}; $B$ denotes batch size.

\begin{table}[t]
  \centering
  \caption{Per-layer Gated DeltaNet tree-verification efficiency on
  Qwen3.5-9B.}
  \label{tab:kernel-verify}
  \footnotesize
  \setlength{\tabcolsep}{4.5pt}
  \begin{tabular}{@{}llrrcrrrr@{}}
    \toprule
    & & \multicolumn{3}{c}{Latency ($\mu$s)}
      & \multicolumn{2}{c}{Occupancy}
      & \multicolumn{2}{c}{L1/shared} \\
    \cmidrule(lr){3-5}\cmidrule(lr){6-7}\cmidrule(lr){8-9}
    Plat. & $B$ & Serial & \sys & \makebox[0pt]{Speedup}
      & Serial & \sys & Serial & \sys \\
    \midrule
    A100 & 1  & 174  & 50.6 & $\mathbf{3.4\times}$
      & 1.9\%  & 15.0\% & 14.7\% & 60.8\% \\
    A100 & 16 & 822  & 110  & $\mathbf{7.5\times}$
      & 11.4\% & 41.4\% & 51.3\% & 93.5\% \\
    GB10 & 1  & 333  & 66.6 & $\mathbf{5.0\times}$
      & 6.6\%  & 43.1\% & 10.1\% & 34.6\% \\
    GB10 & 16 & 3420 & 443  & $\mathbf{7.7\times}$
      & 16.4\% & 58.9\% & 12.5\% & 36.8\% \\
    \bottomrule
  \end{tabular}
\end{table}

Across both platforms, \sys accelerates the complete verification core by
$3.4\times$--$7.7\times$. The speedup rises from $3.4\times$ and
$5.0\times$ at $B=1$ to $7.5\times$ and $7.7\times$ at $B=16$ on A100
and GB10, respectively. This scaling follows from removing the
parent-to-child critical path within every proposal tree, which exposes
node-level parallelism both within and across requests.

The closed-form solver exposes node-level parallelism, converting proposal
nodes into independently schedulable GPU work and increasing achieved
occupancy by $3.6\times$--$7.9\times$ across the four settings. By keeping
the main verification work on chip rather than materializing full recurrent
states in HBM, the factorized, value-tiled path increases L1/shared-memory
throughput by $1.8\times$--$4.1\times$. At $B=16$, it writes 9.8\,MB per
layer instead of 824\,MB, an $84\times$ reduction.
State materialization consumes 38\% and 86\% of serial verification time
on A100 and GB10, respectively, explaining the larger benefit on GB10's
lower-bandwidth unified memory.

Factorization adds one batched commit over only the accepted path. The
commit runs once per target-model forward and accounts for less than
0.5\% of its latency.

\ra[A.Q4]{With batch size four on A100, a draft forward averages 1.07 and
0.85~ms for Qwen3.5-4B and 122B-A10B, respectively, compared with
11.46 and 34.82~ms for a target verification forward. Draft-model
execution, including accepted-path extension, takes 35.8\% and
15.0\% of speculative-round time.}

\begin{table}[t]
  \centering
  \caption{\rb[B.Q2]{GPU energy estimates at batch size one on A100.}}
  \label{tab:eval-energy}
  \footnotesize
  \setlength{\tabcolsep}{2pt}
  \begin{tabular*}{\linewidth}{@{\extracolsep{\fill}}lrrrrrrr@{}}
    \toprule
    & \multicolumn{3}{c}{SGLang-Tree} & \multicolumn{4}{c}{\sys} \\
    \cmidrule(lr){2-4}\cmidrule(lr){5-8}
    Model & tok/s & W & J/token & tok/s & W & J/token & Saving \\
    \midrule
    9B        & 190.10 & 278 & 1.465 & 268.00 & 305 & \textbf{1.137} & \textbf{22.3\%} \\
    122B-A10B & 139.95 & 959 & 6.852 & 179.65 & 1,074 & \textbf{5.977} & \textbf{12.8\%} \\
    \bottomrule
  \end{tabular*}
\end{table}

\rb[B.Q2]{\tabref{tab:eval-energy} reports GPU power measured using NVML~\cite{nvidiaNVML}
for the offline experiments in \secref{sec:eval-offline} at batch size one.
\sys reduces energy per token
by 22.3\% on Qwen3.5-9B and 12.8\% on 122B-A10B.}

\subsection{\cq{Alternative Parallel Verifiers}}
\label{sec:eval-alternatives}

\begin{figure}[t]
  \centering
  \captionsetup[subfloat]{captionskip=2pt}
  \subfloat[\cq{Kernel execution time}\label{fig:eval-verification-alternatives}]{%
    \includegraphics[width=0.49\columnwidth]{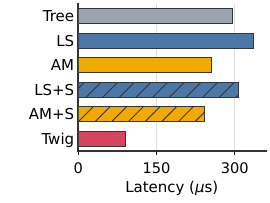}}
  \hfill
  \subfloat[\cq{Decode throughput}\label{fig:eval-verification-throughput}]{%
    \includegraphics[width=0.49\columnwidth]{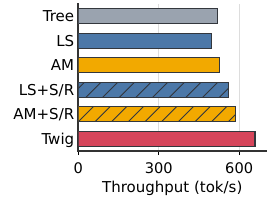}}
  \caption{\cq[CQ1]{Parallel tree-verification alternatives.
  Tree denotes SGLang-Tree, LS denotes layer-sync, and AM denotes
  ancestor-masked sequence-parallel verification. S and R denote compact
  state management and runtime optimizations.}}
  \label{fig:eval-tree-verification}
\end{figure}

\cq[CQ1]{\figref{fig:eval-tree-verification} compares alternative parallel verifiers on
Qwen3.5-9B/A100. Layer-sync (LS) batches independent nodes at each tree depth.
Ancestor-masked verification (AM) adds ancestor masking to Gated DeltaNet's
sequence-parallel prefill kernel~\cite{yang2025gateddeltanet}.
S and R apply \sys's compact state management and runtime optimizations
to both alternatives.}

\cq[CQ1]{\sys accelerates kernel execution by $3.41\times$ over LS+S and
$2.69\times$ over AM+S, with decode gains of 17.9\% and 12.1\% over
their S/R variants. Tree already parallelizes across requests, heads,
and value tiles. In \texttt{ncu} profiles, LS raises time-weighted occupancy
from 6.9\% to 10.1\%, but issued warps per scheduler per cycle fall from
0.23 to 0.21. Additional level boundaries and state transfers outweigh
the extra node parallelism. AM moves intermediates between separate
factor-construction, solve, and readout kernels, while \sys fuses
propagation and readout in one CTA.}

\subsection{\cq{Sensitivity Analysis}}
\label{sec:eval-budget}

\begin{figure}[t]
  \centering
  \captionsetup[subfloat]{captionskip=2pt}
  \subfloat[\cq{Depth}\label{fig:eval-depth}]{%
    \includegraphics[width=0.49\columnwidth]{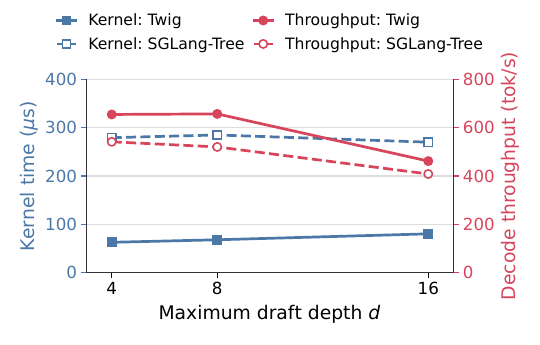}}
  \hfill
  \subfloat[\cq{Top-$k$}\label{fig:eval-topk}]{%
    \includegraphics[width=0.49\columnwidth]{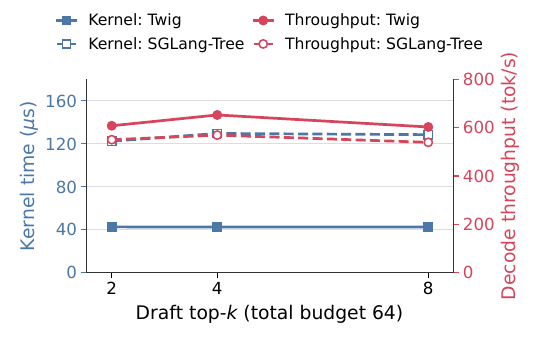}}
  \caption{\cq[CQ3]{Tree-shape sensitivity on Qwen3.5-9B/A100.}}
  \label{fig:eval-tree-shape}
  \vspace{-12pt}
\end{figure}

\begin{figure}[t]
  \centering
  \captionsetup[subfloat]{captionskip=2pt}
  \subfloat[\cq{Node scaling, $B=4$}\label{fig:eval-node-count}]{%
    \includegraphics[width=0.49\columnwidth]{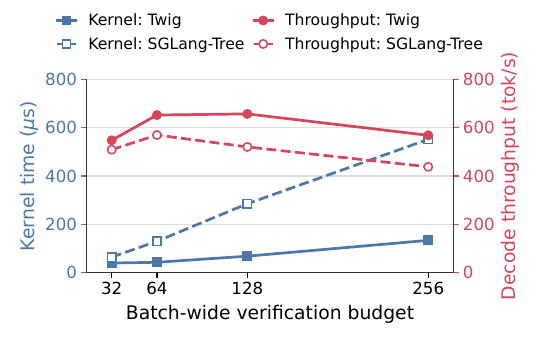}}
  \hfill
  \subfloat[\cq{Platform budgets, $B=8$}\label{fig:eval-budget}]{%
    \includegraphics[width=0.49\columnwidth]{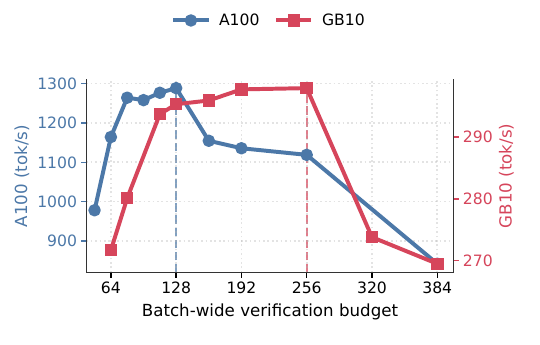}}
  \caption{\cq[CQ3]{Verification-budget sensitivity on Qwen3.5-9B. $B$ denotes batch size.}}
  \label{fig:eval-budget-sensitivity}
\end{figure}

\cq[CQ3]{\figref{fig:eval-tree-shape} evaluates tree-shape sensitivity on
Qwen3.5-9B/A100 at batch size four. With top-$k$ four and total budget 128,
increasing depth from four to sixteen modestly raises \sys's kernel time
from 62.6 to 80.1~$\mu$s, retaining $3.37$--$4.46\times$ kernel speedups
and 13.1--26.3\% higher decode throughput over SGLang-Tree.
Beyond depth eight, drafting cost outweighs additional acceptance,
reducing both systems' throughput.}

\cq[CQ3]{At depth eight and total budget 64, varying $k=2,4,8$ keeps
\sys's kernel time within 42.4--42.6~$\mu$s, with $2.88$--$3.06\times$
speedups over SGLang-Tree, showing stable verification efficiency across
branching settings. Decode throughput peaks at $k=4$ and improves
by 10.5--14.7\% over SGLang-Tree across the sweep.}

\cq[CQ3]{\figref{fig:eval-node-count} varies the total node budget from 32
to 256 at batch size four, depth eight, and top-$k$ four, retaining
$1.62$--$4.20\times$ kernel speedups and 7.7--29.8\% decode gains over
SGLang-Tree. Larger budgets expose more paths until execution cost outweighs
acceptance. At batch size eight (\figref{fig:eval-budget}), throughput peaks
at 128 total nodes on A100 and 256 on GB10, supporting
\figref{fig:motivation}(c) and validating \secref{sec:hybrid-budget}'s design.
Runtime linear-kernel latency deviates from calibration by at most 11.7\%,
and selected budgets perform within 1.2\% of the mixed-batch oracle
end to end.}

\subsection{\cq{Numerical Accuracy}}
\label{sec:eval-accuracy}

\begin{table}[t]
  \centering
  \caption{\cq[CQ4]{Numerical accuracy and task quality on Qwen3.5-9B.}}
  \label{tab:eval-accuracy}
  \small
  {\ifrevisioncolors\color{cCommon}\fi
  \begin{tabular*}{\columnwidth}{@{\extracolsep{\fill}}lrr@{}}
    \toprule
    Metric & SGLang-Tree & \sys \\
    \midrule
    State NRMSE & $1.48\times10^{-7}$ & $4.42\times10^{-7}$ \\
    Output NRMSE & $1.76\times10^{-3}$ & $2.24\times10^{-3}$ \\
    Logit NRMSE & $1.48\times10^{-2}$ & $1.49\times10^{-2}$ \\
    MBPP pass@1 & 68.2\% & 68.4\% \\
    GSM8K accuracy & 95.2\% & 95.6\% \\
    \bottomrule
  \end{tabular*}}
\end{table}

\cq[CQ4]{Following prior work~\cite{shah2024flashattention3}, we measure state
and output errors against an FP64 sequential reference. For logits, only
linear recurrences use FP64 in the reference forward. Normalized root mean
squared error (NRMSE) divides the Euclidean error norm by the reference norm.
\tabref{tab:eval-accuracy} reports maxima across sampled tensors.
\sys and SGLang-Tree exhibit errors of the same order, with state, output,
and logit NRMSE below $5\times10^{-7}$, $2.3\times10^{-3}$, and
$1.5\times10^{-2}$, respectively. State and output bounds hold across
depths 4--16, gate distributions, and model scales 4B/9B/27B.
MBPP code generation and GSM8K mathematical reasoning show comparable
accuracy, with relative differences below 0.5\% from SGLang-Tree.
During 4,096-token generation, paired serial-verification checks find
0.20\% top-1 mismatch across 2,048 candidate nodes and identical
acceptance decisions in all 64 sampled rounds.}

\subsection{Component Ablation}
\label{sec:eval-ablation}

\begin{figure}[t]
  \centering
  \includegraphics[width=0.94\columnwidth]{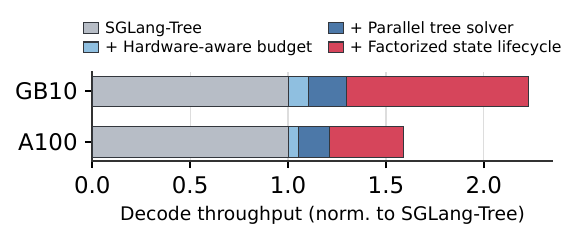}
  \caption{Cumulative online ablation of \sys's three designs over
  SGLang-Tree.}
  \label{fig:eval-ablation}
  \vspace{-12pt}
\end{figure}

\figref{fig:eval-ablation} cumulatively adds \sys's three designs to
SGLang-Tree under the online agent workload. \re[E.Q2]{The ablation starts from
a fixed total proposal budget of 192, midway between the optimal budgets
of 128 on A100 and 256 on GB10.}

Hardware-aware budgeting improves throughput by 5.3\% on A100 and 10.7\%
on GB10 while retaining the serial verifier and state representation.
It stops batch expansion when additional proposal cost outweighs expected
acceptance, isolating the benefit of platform-matched budgets.

Parallel verification raises cumulative gains to $1.21\times$ on A100 and
$1.30\times$ on GB10. Removing parent-to-child serialization exposes
cross-request and intra-tree parallelism among the budgeted branches,
outperforming budgeting alone.

Factorized state management raises cumulative speedups to $1.59\times$
on A100 and $2.23\times$ on GB10. At A100's 128-token budget, it reduces
Qwen3.5-122B-A10B's transient state from 18.0 to 0.18~GiB, freeing 17.8~GiB
for KV cache. GB10 benefits more from its tighter memory capacity and
improved cache hits (\secref{sec:eval-online}). Together, the three designs
select useful proposals, verify them in parallel, and preserve reusable
agent state.

\section{Related Work}
\label{sec:related-work}

\runinhead{Tree speculative decoding.}
Speculative decoding accelerates autoregressive inference by
verifying a sequence of drafted tokens in parallel~\cite{
leviathan2023speculative,chen2023accelerating,Gloeckle2024better,
liu2026speculative,wang2026adaptive,li2024specpim,zhang2026swiftspec}.
Among these methods, tree speculation stands out for longer accepted paths
and greater speedups~\cite{cai2024medusa,li2024eagle2,
li2025eagle3,chen2024sequoia}.
Medusa~\cite{cai2024medusa} and EAGLE-3~\cite{li2025eagle3} generate branches
with learned predictors, while Sequoia~\cite{chen2024sequoia},
EAGLE-2~\cite{li2024eagle2}, Yggdrasil~\cite{guan2025yggdrasil}, and
GTO~\cite{hu2026gto} optimize static or probability-aware tree topologies.
SpecInfer~\cite{miao2024specinfer} and DeFT~\cite{yao2025deft} reduce
target-side verification cost. These works do not transfer to recurrent linear
attention, which exposes no attention matrix~\cite{
miao2024specinfer,yao2025deft}. AdaServe~\cite{li2026adaserve} allocates
proposal budgets across requests to meet SLO requirements.

\runinhead{Hybrid-attention model serving.}
vLLM~\cite{kwon2023efficient}, SGLang~\cite{zheng2024sglang}, and related
systems optimize batching, caching, and GPU execution~\cite{amey2024sarathiserve,
patel2025splitwise,kan2025nanoflow,ruoyu2025mooncake,ye2025flashinfer,
kamath2025podattention}. For hybrid models, Marconi~\cite{pan2025Marconi}
caches recurrent prefixes, while HLX~\cite{Jung2025hlx} and
Pimba~\cite{pim2025Pimba} specialize Transformer--SSM execution.
STree~\cite{wu2025stree} composes diagonal SSM transitions over proposal
trees, but this algebra does not extend to the token-dependent, non-diagonal
recurrence of gated linear attention. \rf[F.Q1]{DeltaNet~\cite{yang2024parallelizing}
and Gated DeltaNet~\cite{yang2025gateddeltanet} use compact triangular
formulations for sequence-parallel computation. \sys differs in three aspects.
First, dependencies follow tree ancestry rather than sequence chunks.
Second, bounded tree depth enables fused finite-Neumann propagation and readout
on GPU value tiles. Third, factorized speculative states share ancestor
updates across branches and reconstruct only the accepted state at commit.}

\section{Conclusion}

\sys combines closed-form parallel verification, factorized accepted-branch
state management, and hardware-aware batch-wide budgeting for hybrid-attention LLMs.
Integrated into SGLang, it accelerates linear-attention verification by up
to $7.7\times$ and reduces transient state memory by $82$--$99\times$.
Against the strongest tree-serving baseline, \sys achieves up to
$2.03\times$ throughput and online-agent latency reductions reaching
$67.6\%$ for TTFT and $49.9\%$ for TPOT.

\ifdefined\hpcacameraready
\section*{Acknowledgments}
This work was supported by Ant Group Research Intern Program.
\fi


\bibliographystyle{IEEEtranS}
\bibliography{refs}

\end{document}